\documentclass[a4paper,11pt]{cmtt-semantics-arxiv}
\pdfoutput=1
\usepackage{a4wide}
\usepackage{palatino}

 \renewenvironment{thebibliography}[1]{%
   \begin{oldthebibliography}{#1}%
     \setlength{\parskip}{2pt}%
     \setlength{\itemsep}{2pt}%
     \fontsize{10}{11}
     \selectfont
 }%
 {%
   \end{oldthebibliography}%
 }

\usepackage{etex} 
\usepackage{tikz}
\definecolor{shade}{HTML}{F0F0F0}
\newcommand\shademath[1]
{
\raisebox{-3pt}{\begin{tikzpicture}
\node [fill=shade,rounded corners=2pt]
{$#1$};
\end{tikzpicture}}
}
\usepackage{float}
\floatstyle{boxed} \restylefloat{figure}

\usepackage{amsfonts}
\usepackage{amsmath}
\usepackage{amsthm}
\usepackage[english]{babel}
\usepackage[colorlinks,urlcolor=black]{hyperref}
\usepackage{prooftree}
\usepackage{stmaryrd}
\usepackage{amssymb} 
\usepackage{graphicx}
\usepackage{mdwlist}

\usepackage{array}
\newcolumntype{L}[1]{>{$}p{#1}<{$}}
\newcolumntype{C}[1]{>{\centering$}p{#1}<{$}}
\newcolumntype{R}[1]{>{\raggedleft$}p{#1}<{$}}
\newcommand\maketab[2]
   {\newenvironment{#1}{\begin{quote}\noindent\begin{tabular}{#2}}{\end{tabular}\end{quote}}
    \newenvironment{#1noquote}{\noindent\begin{tabular}{#2}}{\end{tabular}}
    }

\newtheorem{thrm}{Theorem}[section]
\newtheorem{prop}[thrm]{Proposition}
\newtheorem{lemm}[thrm]{Lemma}
\newtheorem{corr}[thrm]{Corollary}
\newtheorem{_nttn}[thrm]{Notation}
\newenvironment{nttn}{\begin{_nttn}\normalfont}{\end{_nttn}}
\newtheorem{_defn}[thrm]{Definition}
\newenvironment{defn}{\begin{_defn}\normalfont}{\end{_defn}}
\newtheorem{_xmpl}[thrm]{Example}
\newenvironment{xmpl}{\begin{_xmpl}\normalfont}{\end{_xmpl}}
\newtheorem{_rmrk}[thrm]{Remark}
\newenvironment{rmrk}{\begin{_rmrk}\normalfont}{\end{_rmrk}}

\newcommand\hd{\f{hd}}
\newcommand\tl{\f{tl}}
\newcommand\letbox[3]{\f{let\,}#1{=}#2\f{\ in\,}#3}

\newcommand\atoms{{\mathbb A}}
\newcommand\unknowns{{\mathbb X}}

\newcommand\tya{A}
\newcommand\tyb{B}
\newcommand\tyc{C}
\newcommand{\rtm}{r}
\newcommand{\stm}{s}

\newcommand{\fa}{\f{fa}}

\newcommand\bto{\mathrel{\to_\beta}}

\newcommand\cent{\vdash}

\newcommand{\id}{{id}}

\newcommand\fv{\f{fu}}
\newcommand\Forall[1]{\forall #1.}

\newcommand\limp{\Rightarrow}
\newcommand{\fto}{{\to}}
 
\newcommand\dom{{\f{dom}}}

\newcommand\nat{\mathbb N}

\usepackage{graphics} 
\usepackage{type1cm}
\usepackage{eso-pic}

\makeatletter
\AddToShipoutPicture{%
            \setlength{\@tempdimb}{.5\paperwidth}%
            \setlength{\@tempdimc}{.5\paperheight}%
            \setlength{\unitlength}{1pt}%
            \put(\strip@pt\@tempdimb,\strip@pt\@tempdimc){%
        \makebox(0,0){\rotatebox{55}{\textcolor[gray]{0.98}%
        {\fontsize{5cm}{5cm}\selectfont{}}}}%
            }%
}
\makeatother

\usepackage{fancybox}
\newlength{\mylength}
\newenvironment{frameqn}%
{\setlength{\fboxsep}{6pt}
\setlength{\mylength}{\linewidth}%
\addtolength{\mylength}{-2\fboxsep}%
\addtolength{\mylength}{-2\fboxrule}%
\Sbox
\minipage{\mylength}%
\setlength{\abovedisplayskip}{0pt}%
\setlength{\belowdisplayskip}{0pt}%
$$}%
{$$\endminipage\endSbox
\[\fbox{\TheSbox}\]}

\newenvironment{frametxt}%
{\setlength{\fboxsep}{5pt}
\setlength{\mylength}{\linewidth}%
\addtolength{\mylength}{-2\fboxsep}%
\addtolength{\mylength}{-2\fboxrule}%
\Sbox
\minipage{\mylength}%
\setlength{\abovedisplayskip}{0pt}%
\setlength{\belowdisplayskip}{0pt}%
}%
{\endminipage\endSbox
\[\fbox{\TheSbox}\]}

\usepackage{marvosym}\newcommand\at{\MVAt}  
\newcommand{\act}{{\cdot}}

\newcommand{\deffont}[1]{\textbf{#1}}

\newcommand{\f}[1]{\ensuremath{\mathit{#1}}}

\newcommand{\lam}[1]{\lambda{#1}.}

\newcommand{\rulefont}[1]{\ensuremath{(\mathbf{#1})}}

\newcommand{\ssm}{:=}
\newcommand{\tf}[1]{\mathsf{#1}}

\newcommand{\BBox}{{\boxdot}}

\newcommand{\denot}[3]{\llbracket #3 \rrbracket_{\scalebox{.8}{$#2$}}^\den{#1}} 
\newcommand{\hdenot}[2]{{\denot{}{#1}{#2}}}

\newcommand\den[1]{{\hspace{.00ex}\scalebox{.55}{$#1$}}}
\newcommand\hden{\den{\interp H}} 

\usepackage[mathscr]{eucal}
\newcommand\interp[1]{\ensuremath{\mathscr #1}}

\author{\href{http://www.gabbay.org.uk}{Murdoch J. Gabbay} and \href{http://software.imdea.org/~aleks/}{Aleksandar Nanevski}}
\title{Denotation of syntax and metaprogramming in contextual modal type theory (CMTT)}
\date{}

\begin{document}

\begin{abstract}
The modal logic S4 can be used via a Curry-Howard style correspondence to obtain a $\lambda$-calculus. 
Modal (boxed) types are intuitively interpreted as `closed syntax of the calculus'.
This $\lambda$-calculus is called modal type theory --- this is the basic case of a more general \emph{contextual} modal type theory, or CMTT. 

CMTT has never been given a denotational semantics in which modal types \emph{are} given denotation as closed syntax.
We show how this can indeed be done, with a twist.
We also use the denotation to prove some properties of the system.
\\[.5ex]


\noindent \emph{MSC-class:} 03B70 (logic in computer science); 03B45 (modal logic); 68Q55 (semantics)

\noindent \emph{AMS-class:} F.4.1 (modal logic); F.3.2 (semantics of programming languages)
\end{abstract}

\maketitle

\tableofcontents
\section{Introduction}

The box modality $\Box$ from modal logic has proven its usefulness in logic.
It admits various logical and semantic interpretations in the spirit of `we know that' or `we can prove that' or `in the future it will be the case that'.
A nice historical overview of modal logic, which also considers the specific impact of computer science, is in \cite[Subsection~1.7]{blackburn:modl}. 

CMTT (contextual modal type theory) is a typed $\lambda$-calculus based via the Curry-Howard correspondence on the modal logic S4.
The box modality becomes a type-former, and box types are intuitively interpreted as `closed syntax of'.

So CMTT has types for programs that generate CMTT syntax.

Because of this, CMTT has been applied to meta-programming,
but it has independent interest as a language, designed according to rigorous mathematical principles and in harmony with modal logic, which interprets $\Box$ in a programming rather than a logical context.
Box types are types of the syntax of terms. 

Until now this has not been backed up by a denotational semantics in which box types really \emph{are} populated by the syntax of terms.
In this paper, we do that: our intuitions are realised in the denotational semantics in a direct and natural, and also unexpected, manner.

The denotation is interesting from the point of view of the interface between logic and programming.
Furthermore, we exploit the denotation to prove properties of the language, showing how denotations are not only illuminating but can also serve for new proof-methods.

\subsection{Keeping it simple}

This paper considers two related systems:
\begin{itemize*}
\item
The purely modal system, based on box types like $\Box\tya$.
\item
The contextual modal system, based on `boxes containing types' like $[\tya_1,\tya_2]\tyb$---the reader might like to think of the contextual system as a multimodal logic \cite[Subsection~1.4]{gabbay:mandml} (whose modalities are themselves indexed over propositions).
\end{itemize*}
Broadly speaking, the purely modal system is nicer to study but a little too simple. 
The contextual modal system generalises the purely modal system and gives it slightly more expressive power, but it can be a little complicated; not obscure, just long to write out.

Therefore, we open this paper with the modal system, make the main point of our denotation in the simplest and clearest possible manner---the reader who wants to jump right in and work backwards could do worse than start with the example denotations in Subsection~\ref{subsect.den.letbox} onwards---and then we consider the contextual system as the maths becomes more advanced.
Section~\ref{sect.box.syntax} presents syntax and typing of the modal system and Section~\ref{sect.contextual.system} does the same for the contextual modal system; Section~\ref{sect.denotations} gives modal denotations and Section~\ref{sect.contextual.models} gives contextual modal denotations.

The developments are parallel, but not identical.
Where proofs are not very different between the modal and contextual systems, we omit routine repetition. 
We consider reduction of the modal system in Section~\ref{sect.reduction} but not reduction of the contextual system. 
Also, we develop the important notion of \emph{shapeliness} only for the contextual system in Section~\ref{sect.shapely}; it is obvious how the modal case would be a special case.

\subsection{Key ideas}

Our main technical results are Theorems~\ref{thrm.type.soundness} and~\ref{thrm.type.soundness.c}, and Corollary~\ref{corr.more.c}.

However, just looking at these results may be misleading; the key technical ideas that make these results work, and indeed contribute to making them interesting, occur beforehand. 

So it might be useful to list some of the key ideas in the paper.
This list is not an exhaustive technical overview, so much as clues for the reader who wants to gain some quick insight and navigate the mathematics.
Here are some of the main points that make the mathematics in this paper different and distinctive:
\begin{itemize*}
\item
\emph{Inflation} in the case of $\hdenot{}{\Box\tya}$ in Figure~\ref{fig.denot.types}, and the `tail of' semantics of $X_\at$ in Figure~\ref{fig.denot.terms}.
This is discussed in Remark~\ref{rmrk.why.inflate}.
\item
Proposition~\ref{prop.rtheta} and the fact that it is needed for soundness of the denotation.
\item
The remarkable Proposition~\ref{prop.typesub}, in which valuations get turned into substitutions and closed syntax in the denotation interacts directly with the typing system.
This is a kind of dual to the interaction seen in Proposition~\ref{prop.rtheta}. 
\item
The denotation of $\hdenot{}{[\tya_i]\tya}$ in Figure~\ref{fig.denot.types.c}, which in the context of the rest of the paper is very natural.
\item
The notion of \emph{shapeliness} in Definition~\ref{defn.shape} and the `soundness result' Proposition~\ref{prop.shape.r}.
\end{itemize*}
We discuss all of these in the body of the paper.

\subsection{On intuitions}

\subsubsection{`Syntax' means syntax}

One early difficulty the authors of this paper faced was in communication, because we sometimes used terms synonymously without realising that the words were so slippery.

The intuition we give to $\Box\tya$ is self-reflectively \emph{closed syntax of the language itself}.
This is a distinct intuition from `computations', `code', `values', or `intensions', because these are not necessarily intended self-reflectively.

It is very important not to confuse this intuition with apparently similar intuitions expressed as `code of $\tya$', `values of $\tya$', `computations of $\tya$', or `intension of $\tya$'.
These are not quite the same thing.
It may be useful to briefly survey them here: 
\begin{itemize*}
\item
`Code of $\tya$' is an ambiguous term; this is often understood as precompiled code or bytecode, rather than syntax of the original language. 
See \cite{wic+lee+pfe:pldi98} for a system based on that intuition.
\item
`Values of $\tya$' is a dangerous intuition and there probably should be a law against it: depending on whom one is speaking with, this could be synonymous in their mind with `normal forms of $\tya$' (a syntactic notion) or `denotations of $\tya$' (a non-syntactic notion).

Matters become even worse if one's interlocuteur assumes that denotations may be silently added to syntax as constants (fine for mathematicians; not so fine for programmers).
More than one conversation has been corrupted by the associated misunderstandings.
\item
For a discussion of `computation of $\tya$' see the Related Work in the Conclusions, where we discuss how this intuition can lead to a notion of Moggi-style \emph{monad}.
\item
`Intension of $\tya$' is similar to `syntax of $\tya$', but significantly more general: there is no requirement that the intension be syntactic, or if it is syntactic, that it be the same calculus.
One could argue that `intension of' should also satisfy that the denotation of $\Box\Box\tya$ be identical in some strong sense---e.g. be the same set as---to that of $\Box\tya$, since taking an intension twice should reveal no further internal structure.
(This does not match the denotation of this paper.)

An interesting (and as far as we know unexplored) model of this intuition might be partial equivalence relations (\deffont{PER}s), where $\Box\tya$ takes $\tya$ and forms the identity PER which is defined where $\tya$ is defined.\footnote{Alex Simpson and Paul Levy both independently suggested PERs when the first author sketched the ideas of this paper, and Simpson went further and suggested the specific model discussed above.  We are grateful to Levy and Simpson for their comments, which prompted us to be specific about the intuition behind the particular denotation in this paper.}
Famously, PERs form a cartesian-closed category \cite[Subsection~3.4.1]{asperti:cattsi}.
\end{itemize*}

In short: where the reader sees `$\Box\tya$', they should think `raw syntax in type $\tya$'. 

\subsubsection{`Functions' means functions}

It may be useful now to head off another possible confusion: 
where the reader sees $\tya\fto\tyb$, they should think `graph of a function'---not `computable function', `representable function', `syntax of a function', or `code of a function'.

All of these things are also possible, but in this paper our challence is to create a type system, language, and denotation which are `epsilon away' from the simply-typed $\lambda$-calculus or (since we admit a type of truth-values) higher-order logic---and it just so happens that we also have modal types making precisely its \emph{own} syntax into first-class data.

So: we are considering a `foundations-flavoured' theory in which $\tya\fto\tyb$ represents all possible functions (in whatever foundation the reader prefers) from $\tya$ to $\tyb$, and we do not intend this paper to be `programming-flavoured' in which $\tya\fto\tyb$ represents only that function(-code) or normal forms that can exist inside some computational device.
And, $\Box\tya$ should represent, as much as possible, `the syntax of our language/logic that types as $\tya$'.

\section{Syntax and typing of the system with box types}
\label{sect.box.syntax}

We start by presenting the types, terms, and typing relation for the modal type system.
This is the simplest version of the language that we want to give a denotational semantics for.

\subsection{The basic syntax}

\begin{defn}
\label{defn.atoms}
Fix two countably infinite sets of \deffont{variables} $\atoms$ and $\unknowns$.
We will observe a \deffont{permutative convention} that $a,b,c,\dots$ will range over distinct variables in $\atoms$ and $X,Y,Z,\dots$ will range over distinct variables in $\unknowns$.
We call $a,b,c$ \deffont{atoms} and $X,Y,Z$ \deffont{unknowns}. 
\end{defn}

\begin{defn}
\label{defn.types}
Define \deffont{types} inductively by:
\begin{frameqn}
\begin{array}{r@{\ }l}
\tya::=& o \mid \nat \mid \tya\fto \tya \mid \Box \tya
\end{array}
\end{frameqn}
\end{defn}

\begin{nttn}
By convention, if $X$ and $Y$ are sets we will write $Y^X$ for the set of functions from $X$ to $Y$.
This is to avoid any possible confusion between $\tya\fto\tyb$ (which is a type) and $Y^X$ (which is a set).
\end{nttn}

\begin{rmrk}
\begin{itemize*}
\item
$o$ will be a type of \deffont{truth values}; its denotation will be populated by truth-values $\{\bot,\top\}$.
\item
$\nat$ will be a type of \deffont{natural numbers}; its denotation will be populated by numbers $\{0,1,2,\dots\}$.
\item
$\tya\fto\tyb$ is a \deffont{function type}; its denotation will be populated by functions.
\item
$\Box\tya$ is a \deffont{modal type}; its denotation will be populated by syntax.
\end{itemize*}
\end{rmrk}

\begin{defn}
Fix a set of \deffont{constants} $C$ to each of which is assigned a type $\f{type}(C)$.
We write $C:\tya$ as shorthand for `$C$ is a constant and $\f{type}(C)=\tya$'.
We insist that constants include the following:
$$
\begin{array}{c@{\qquad}c@{\qquad}c}
\bot:o
&
\top:o
&
\tf{isapp}_\tya:(\Box \tya)\fto o
\end{array}
$$
We may also assume constants for $\nat$, such as $0:\nat$, $\tf{succ}:\nat\fto\nat$, $*:\nat\fto \nat\fto \nat$ and $+:\nat\fto \nat\fto \nat$, a fixedpoint combinator, we may write $1$ for $\tf{succ}(0)$,\ and so on.\footnote{\dots so we follow the example of PCF \cite{mitchell:foupl}.}

We may omit type subscripts where they are clear from context or do not matter.
\end{defn}

\begin{defn}
\label{defn.terms}
Define \deffont{terms} inductively by:
\begin{frameqn}
r::= C \mid a \mid X_\at \mid \lam{a{:}\tya}r \mid rr \mid \Box r \mid \letbox{X}{r}{r}
\end{frameqn}
\end{defn}
Constants $C$ are, as standard in the $\lambda$-calculus, added as desired to represent logic and computational primitives.
An atom $a$ plays the role of a standard $\lambda$-calculus variable; it is $\lambda$-abstracted in a typed manner in $\lam{a{:}\tya}r$.
The term $X_\at$ means intuitively `evaluate $X$' and $\Box r$ means intuitively `the syntax $r$ considered itself in the denotation'.
Finally $\letbox{X}{s}{r}$ means intuitively `set $X$ to be the syntax calculated by $s$, in $r$'.
Examples of this in action are given and discussed in Subsection~\ref{subsect.some.programs}.

\begin{rmrk}
The effect of $r_\at$ (which is not syntax) is obtained by $\letbox{X}{r}{X_\at}$.
Likewise the effect of $\lam{X{:}\Box\tya}r$ (which is not syntax) is obtained by $\lam{a{:}\Box\tya}\letbox{X}{a}{r}$. 

We \emph{cannot} emulate $\letbox{X}{s}{X_\at}$ using $(\lam{a{:}\tya}a_\at)r$.
The expression `$a_\at$' would mean `evaluate the syntax $a$' rather than `evaluate the syntax linked to $a$'.\footnote{In addition even if $a_\at$ were syntax, it would not type in the typing system of Figure~\ref{fig.modal.types}, because $\fa(a_\at)$ would be equal to $\{a\}\neq\varnothing$ (Definition~\ref{defn.hole.free.atoms}). Modal types are inhabited by \emph{closed} syntax (Definition~\ref{defn.hole.free.atoms}).}
\end{rmrk}

\begin{defn}
\label{defn.hole.free.atoms}
Define \deffont{free atoms} $\fa(\rtm)$ and \deffont{free unknowns} $\fv(\rtm)$ by:
\begin{displaymath}
\begin{array}{r@{\ }l@{\qquad}r@{\ }l}
\fa(C) = & \varnothing
&
\fa(a) = & \{ a \}                               
\\
\fa(\lam{a{:}\tya}r) = & \fa(r) \setminus \{ a \}
&
\fa(\rtm\stm) = & \fa(\rtm) \cup \fa(\stm) 
\\
\fa(\Box r) = & \fa(r)
&
\fa(\letbox{X}{s}{r}) = & \fa(r)\cup\fa(s)
\\
\fa(X_\at) = & \varnothing 
\\[2ex]
\fv(C) = & \varnothing
&
\fv(a) = & \varnothing
\\
\fv(\lam{a{:}\tya}\stm) = & \fv(\stm) 
&
\fv(\rtm\stm) = & \fv(\rtm) \cup \fv(\stm)  
\\
\fv(\Box r) = & \fv(r)
&
\fv(\letbox{X}{s}{r}) = & (\fv(r){\setminus}\{X\})\cup\fv(s)
\\
\fv(X_\at) = & \{ X \}
\end{array}
\end{displaymath}
If $\fa(r)\cup\fv(r)=\varnothing$ then we call $r$ \deffont{closed}.
\end{defn}

\begin{defn}
We take $a$ to be bound in $r$ in $\lam{a{:}\tya}r$ and $X$ to be bound in $r$ in $\letbox{X}{s}{r}$, and we take syntax up to $\alpha$-equivalence as usual.
We omit definitions but give examples:
\begin{itemize*}
\item
$\lam{a{:}\tya}a=\lam{b{:}\tya}b$.
\item
$\lam{a{:}\tya}(X_\at a)=\lam{b{:}\tya}(X_\at b)$.
\item
$\letbox{X}{\Box a}{X_\at b}=\letbox{Y}{\Box a}{Y_\at b}$.
\end{itemize*}
As the use of an equality symbol above suggests, we identify terms up
to $\alpha$-equivalence.\footnote{Using nominal abstract syntax
  \cite{gabbay:newaas-jv} this identification can be made consistent
  with the use of names for bound atoms \emph{and} the inductive
  definition in Definition~\ref{defn.terms}.  However, studying how
  best to define syntax is not the emphasis of this paper.}  
\end{defn}

\subsection{Typing}

\begin{defn}
\label{defn.typing.rules}
\begin{itemize*}
\item
A \deffont{typing} is a pair $a:\tya$ or $X:\Box \tya$.  
\item
A
\deffont{typing context} $\Gamma$ is a finite partial function from
$\atoms\cup\unknowns$ to types.  
\item
A \deffont{typing sequent} is a tuple $\Gamma\cent r:\tya$ of a typing
context, a term, and a type.  
\end{itemize*}
We use list notation for typing contexts, e.g. $a{:}\tya,Y{:}\tyb$ is the function mapping $a$ to $A$ and $Y$ to $B$; and $a{:}\tya\in\Gamma$ means that $\Gamma(a)$ is defined and $\Gamma(a)=\tya$.

\begin{frametxt}
Define the \deffont{valid typing sequents} of the modal type system inductively by the rules in Figure~\ref{fig.modal.types}.
\end{frametxt}
\end{defn}
We discuss examples of typable terms in Subsection~\ref{subsect.some.programs}.
The important rule is \rulefont{\Box I}, which tells us that if we have some syntax $r$ and it has no free atoms, then we can box it as a denotation $\Box r$ of box type---any free unknowns $X$ in $r$/$\Box r$ get linked to further boxed syntax, which is expressed by \rulefont{\Box E}.

\begin{figure}[t]
$$
\begin{array}{c@{\qquad}c}
\begin{prooftree}
\phantom{h}
\justifies
\Gamma,a:\tya\cent a:\tya
\using\rulefont{Hyp}
\end{prooftree}
&
\begin{prooftree}
\phantom{h}
\justifies
\Gamma\cent C:\f{type}(C)
\using\rulefont{Const}
\end{prooftree}
\\[2em]
\begin{prooftree}
\Gamma,a{:}\tya\cent r:\tyb 
\justifies
\Gamma\cent (\lam{a{:}\tya}r):\tya\to \tyb
\using\rulefont{{\to}I}
\end{prooftree}
&
\begin{prooftree}
\Gamma\cent r':\tya\to \tyb
\quad
\Gamma\cent r:\tya
\justifies
\Gamma\cent r'r:\tyb
\using\rulefont{{\to}E}
\end{prooftree}
\\[2em]
\begin{prooftree}
\Gamma\cent r:\tya 
\quad (\fa(r){=}\varnothing)
\justifies
\Gamma\cent \Box r:\Box \tya
\using\rulefont{{\Box}I}
\end{prooftree}
&
\begin{prooftree}
\Gamma\cent s{:}\Box\tya 
\quad 
\Gamma,X{:}\Box\tya\cent r {:} \tyb 
\justifies
\Gamma\cent \letbox{X}{s}{r}:\tyb
\using\rulefont{{\Box}E}
\end{prooftree}
\\[2em]
\begin{prooftree}
\phantom{h}
\justifies
\Gamma,X:\Box\tya \cent X_\at:\tya
\using\rulefont{Ext}
\end{prooftree}
\end{array}
$$
\caption{Modal type theory typing rules}
\label{fig.modal.types}
\end{figure}

\begin{nttn}
We may write $\varnothing\cent r:\tya$ just as $r:\tya$.
\end{nttn}

\begin{nttn}
\label{nttn.restrict}
If $\Gamma$ is a typing context and $U\subseteq\mathbb A\cup\mathbb X$ then write $\Gamma|_U$ for $\Gamma$ \deffont{restricted} to $U$.
This is the partial function which is equal to $\Gamma$ where it is defined, and $\dom(\Gamma|_U)=\dom(\Gamma)\cap U$.
\end{nttn}
 
Proposition~\ref{prop.ws} combines Weakening and Strengthening: 
\begin{prop}
\label{prop.ws}
If $\Gamma\cent r:\tya$ and $\Gamma'|_{\fv(r)\cup\fa(r)}=\Gamma|_{\fv(r)\cup\fa(r)}$ then $\Gamma'\cent r:\tya$.
\end{prop}
\begin{proof}
By a routine induction on $r$.
\end{proof}

\maketab{tdef}{@{\hspace{-3em}}R{4em}@{\ }L{40em}}

\subsection{Examples of terms typable in the modal system}
\label{subsect.some.programs}

We are now ready to discuss intuitions about this syntax; for a more formal treatment see Section~\ref{sect.denotations} which develops the denotational semantics.
We start with some short examples and then consider more complex terms.

\subsubsection{Short examples}
\label{subsubsection.one.line}

\begin{enumerate}
\item
Assume constants $\neg:o\fto o$ and $\land:o\fto o\fto o$, where $\land$ is written infix as usual.
Then we can type 
\begin{tdef}
\varnothing \cent&\lam{a{:}\Box o}\letbox{X}{a}{\Box(\neg X_\at)}\ :\ \Box o\fto\Box o.
\\
\varnothing \cent&\lam{a{:}\Box o}\lam{b{:}\Box o}\letbox{X}{a}{\letbox{Y}{b}{\Box(X_\at\land Y_\at)}}\ :\ \Box o\fto\Box o\fto\Box o.
\\
\varnothing \cent&\lam{a{:}\Box o}\letbox{X}{a}{\Box(X_\at\land X_\at)}\ :\ \Box o\fto\Box o.
\end{tdef}
Intuitively these represents the syntax transformations $P\mapsto\neg P$,\ $P,Q\mapsto P\land Q$,\ and $P\mapsto P\land P$.
\item
\label{item.T}
This program takes syntax of type $\tya$ and evaluates it:
\begin{tdef}
\varnothing\cent& \lam{a{:}\Box\tya}\letbox{X}{a}{X_\at}\ :\ \Box\tya\fto\tya
\end{tdef}
This corresponds to the modal logic axiom \rulefont{T}.
\item
Expanding on the previous example, this program takes syntax for a function and an argument, evaluates the syntax and applies the function to the argument:
\begin{tdef}
\varnothing\cent& \lam{a{:}\Box(\tya\fto\tyb)}\lam{b{:}\tya}(\letbox{X}{a}{X_\at})b\ :\ \Box(\tya\fto\tyb)\to(\tya\fto\tyb)
\end{tdef}
\item
\label{item.4}
This program takes syntax of type $\tya$ tagged with $\Box$, and adds an extra $\Box$ so that it becomes syntax of type $\Box\tya$:
\begin{tdef}
\varnothing\cent& \lam{a{:}\Box\tya}\letbox{X}{a}{\Box\Box X_\at}\ :\ \Box\tya\fto\Box\Box\tya
\end{tdef}
This corresponds to the modal logic axiom \rulefont{4}.
\end{enumerate}

\subsubsection{There is no natural term of type $\tya\fto \Box \tya$}
\label{subsubsect.to.box}

We can try to give $\lam{a{:}o}\Box a$ the type $\tya\fto\Box\tya$, but we fail because the typing
context $a{:}o$ does not satisfy $\fa(a)=\varnothing$.

Our denotation of Figures~\ref{fig.denot.types} and~\ref{fig.denot.terms} illustrates that it is not in general possible to invert the evaluation map from Subsection~\ref{subsubsection.one.line} and thus map $\tya$ to $\Box\tya$.
This is Corollary~\ref{corr.later}.\footnote{For sufficiently `small' types this may be possible by specific constructions; see Example~\ref{xmpl.unevaluate.nat}.}
So
\begin{itemize*}
\item
there is a canonical map $\Box\tya\fto\tya$ (syntax to denotation)---we saw this map in part~1 of this example---but 
\item
not in general an inverse map $\tya\fto\Box\tya$ (denotation to syntax).
\end{itemize*}

\subsubsection{A term for \emph{Axiom K}}
\label{subsubsect.K}

\emph{Axiom K}, also called the \emph{normality axiom} \cite[Definition~1.39, Subsection~1.6]{blackburn:modl}; its type is $\Box(\tya\fto\tyb)\fto\Box\tya\fto\Box\tyb$.

We can write a term of this type.
Intuitively, the term below takes syntax for a function and syntax for an argument, and produces syntax for the function applied to the argument:
$$
\begin{array}{r@{\ }l@{\ }l}
\varnothing \cent& \lam{a{:}\Box(\tya \fto\tyb)}\lam{b{:}\Box\tya}\letbox{Y}{b}{\letbox{X}{a}{\Box(X_\at Y_\at)}}
:\Box(\tya \fto\tyb)\fto\Box\tya\fto\Box\tyb
\end{array}
$$

\begin{rmrk}
\label{rmrk.modal}
We exhibited terms of type $\Box\tya\fto\tya$, $\Box\tya\fto\Box\Box\tya$, and $\Box(\tya\fto\tyb)\fto\Box\tya\fto\Box\tyb$, so Figure~\ref{fig.modal.types} implements (at least) the deductive power of an intuitionistic variant of S4 \cite[Subsection~4.1, page~194]{blackburn:modl}.\footnote{The list of axioms of \cite[page~194]{blackburn:modl} uses $\Diamond$ instead of $\Box$.

A most remarkable family of theorems of Kripke semantics for modal logic relates geometric properties of the Kripke frame's \emph{accessibility relation} with logical properties of the modalities.  Axiom \rulefont{K} is satisfied by all frames.  Axiom \rulefont{T} expresses geometrically that accessibility is reflexive.  Axiom \rulefont{4} expresses that accessibility is transitive.}

The reader familiar with category theory may also ask whether $\Box$ can be viewed as a \emph{comonad}, since $\Box\tya\fto\tya$ and $\Box\tya\fto\Box\Box\tya$ look like the types of a \emph{counit} and \emph{comultiplication} (and perhaps $\Box(\tya\fto\tyb)\fto\Box\tya\fto\Box\tyb$ looks like the action of a functor).
We return to this in Section~\ref{sect.box.comonad}.
\end{rmrk}

\subsubsection{The example of exponentiation}
\label{subsubsect.exp}

This is a classic example of meta-programming: write a function that takes a number $n$ and returns syntax for the function $x\in\nat\mapsto x^n$.

Assuming a combinator for
primitive recursion over natural numbers and using some standard sugar, the following term implements exponentiation:
$$
\begin{array}[t]{r@{\ }l}
\f{exp}\ 0\Rightarrow&\Box\lam{b{:}\nat}1 
\\ 
\f{exp}\ (\f{succ}(n))\Rightarrow&\letbox{X}{\f{exp}\,n}{\bigl(\Box \lam{b{:}\nat}b*(X_\at b) \bigr)} . 
\end{array}
$$
However, the term above generates $\beta$-reducts. 
The reader can see this because of the `$\Box \lam{b{:}\nat}b*(X_\at b)$' above.
This application $X_\at b$ is trapped under a $\Box$ and \emph{will not} reduce.
 
Looking ahead to the reduction relation in Figure~\ref{fig.reduction}, $\f{exp}\,2$ reduces to 
$$
\Box(\lam{b{:}\nat}b*(\lam{b{:}\nat}b*((\lam{b{:}\nat}1)b)b))
\quad\text{and not to}\quad
\Box(\lam{b{:}\nat}(b*b*1)) .
$$
Looking ahead to the denotation of Figure~\ref{fig.denot.terms}, the denotation of $\f{exp}\,2$ will likewise be $\Box(\lam{b{:}\nat}b*(\lam{b{:}\nat}b*((\lam{b{:}\nat}1)b)b))$ in a suitable sense.
We indicate the calculation in Subsection~\ref{subsubsect.exp.denot}.

The contextual system of Section~\ref{sect.contextual.system} deals with this particular issue; see Subsection~\ref{subsubsect.exp.c}.

\subsection{Substitution}

\begin{defn}
\label{defn.substitution}
An \deffont{(atoms-)substitution} $\sigma$ is a finite partial function from atoms $\mathbb A$ to terms.
$\sigma$ will range over atoms-substitutions.

Write $\dom(\sigma)$ for the set ${\{a\mid\sigma(a)\text{ defined}\}}$ 

Write $\id$ for the \deffont{identity} substitution, such that $\dom(\sigma)=\varnothing$.
 
Write $[a{\ssm}t]$ for the map taking $a$ to $t$ and undefined elsewhere.

An \deffont{(unknowns-)substitution} $\theta$ is a finite partial function from unknowns $\unknowns$ to terms such that for every $X$, if $X\in\dom(\theta)$ then $\theta(X)=\Box r$ for some $r$ with $\fa(r)=\varnothing$. 

$\theta$ will range over unknowns-substitutions.

We write $\dom(\theta)$,\ $\id$,\ and $[X{\ssm}t]$ just as for atoms-substitutions.
\end{defn}

\begin{defn}
\label{defn.fa.sigma}
Define 
$$
\begin{array}{r@{\ }l}
\fa(\sigma)=&\dom(\sigma)\cup\{\fa(\sigma(a))\mid a\in\dom(\sigma)\}
\quad\text{and}
\\
\fv(\theta)=&\dom(\theta)\cup\{\fv(\theta(X))\mid X\in\dom(\theta)\} .
\end{array}
$$
\end{defn}

\begin{rmrk}
\label{rmrk.theta.varnothing}
Where $\theta$ is defined, it maps $X$ specifically to terms the form $\Box r$ with $\fa(r)=\varnothing$.

This is because `$\Box r$ with $\fa(r)=\varnothing$' is the syntax inhabiting modal types.
If we consider another class of syntax (e.g. in the contextual system of Section~\ref{sect.contextual.system} onwards), then the corresponding notion of unknowns-substitution changes in concert with that.
\end{rmrk}

\begin{figure}[t]
$$
\begin{array}{r@{\ }l@{\qquad}r@{\ }l@{\qquad}l}
C\sigma=&C
&
a\sigma=&\sigma(a) &(a\in\dom(\sigma))
\\
(rs)\sigma=&(r\sigma)(s\sigma)
&
a\sigma=&a &(a\not\in\dom(\sigma))
\\
(\Box r)\sigma=&\Box (r\sigma) 
&
(\lam{c{:}\tya}r)\sigma=&\lam{c{:}\tya}(r\sigma) &
(c\not\in\fa(\sigma))
\\
X_\at\sigma=&X_\at
&(\letbox{Y}{s}{r})\sigma=&\letbox{Y}{s\sigma}{r\sigma} 
\\[2.5ex]
C\theta=&C
&
a\theta=&a
\\
(rs)\theta=&(r\theta)(s\theta)
&
X_\at\theta=&s' 
&(\theta(X)=\Box s')
\\
(\Box r)\theta=&\Box(r\theta)
&
X_\at\theta=&X_\at 
&(X\not\in\dom(\theta))
\\
(\lam{c{:}\tya}r)\theta=&\lam{c{:}\tya}(r\theta) 
&
(\letbox{Y}{s}{r})\theta=&\letbox{Y}{s\theta}{r\theta}
&(Y\not\in\fv(\theta))
\end{array}
$$
\caption{Substitution actions for atoms and unknowns}
\label{fig.sub}
\end{figure}

Definition~\ref{defn.sub} describes how atoms and unknowns get instantiated.
We discuss it in Remark~\ref{rmrk.comments.on.sub} but one point is important above all others: if $\theta(X)=\Box s'$ then $X_\at\theta$ is equal to $s'$.
So a very simple reduction/computation is `built in' to the substitution action for unknowns, that $(\Box s')_\at \to s'$.\footnote{$(\Box s')_\at$ is not actually syntax, but if it were, then $(\Box s')_\at \to s'$ would be its reduction.} 
\begin{defn}
\label{defn.sub}
Define \deffont{atoms} and \deffont{unknowns} substitution actions $r\sigma$ and $r\theta$ inductively by the rules in Figure~\ref{fig.sub}.
\end{defn}

Lemma~\ref{lemm.fa.rtheta} illustrates a nice corollary of the point discussed in Remark~\ref{rmrk.theta.varnothing}.
It will be useful later in Proposition~\ref{prop.typesub}.
\begin{lemm}
\label{lemm.fa.rtheta}
$\fa(r\theta)=\fa(r)$.
\end{lemm}
\begin{proof}
By a routine induction on $r$ using our assumption of Definition~\ref{defn.substitution} that if $X\in\dom(\theta)$ then $\fa(\theta(X))=\varnothing$.
\end{proof}

\begin{rmrk}
\label{rmrk.comments.on.sub}
A few comments on Definition~\ref{defn.sub}:
\begin{itemize*}
\item
The two capture avoidance side-conditions $c\not\in\fa(\sigma)$ and $Y\not\in\fv(\theta)$ can always be guaranteed by renaming.
\item
We write $(\Box r)\sigma=\Box(r\sigma)$.
This is computationally wasteful in the sense that the side-condition $\fa(r)=\varnothing$ on \rulefont{\Box I} (Figure~\ref{fig.modal.types}) guarantees that for typable terms (which is what we care about) $r\sigma=r$.

We prefer to keep basic definitions orthogonal from such optimisations, but
this is purely a design choice (and see the next item in this list).
\item
We write $(\lam{c{:}\tya}r)\theta=\lam{c{:}\tya}(r\theta)$ without any side-condition that $c$ should avoid capture by atoms in $\theta$.
This is because Definition~\ref{defn.substitution} insists that $\fa(\theta(X))=\varnothing$ always, so there can be no capture to avoid.
\end{itemize*}
\end{rmrk}

Recall the definition of $[a{\ssm}s]$ from Definition~\ref{defn.substitution}.
Lemma~\ref{lemm.standard.fv.sub} is a standard lemma which will be useful later:
\begin{lemm}
\label{lemm.standard.fv.sub}
If $a\not\in\fa(r)$ then $r[a{\ssm}s]=r$. 
\end{lemm}
\begin{proof}
By a routine induction on $r$.
\end{proof}

Definition~\ref{defn.gamma.theta} and Proposition~\ref{prop.rtheta} are needed for Proposition~\ref{prop.typesub}.
\begin{defn}
\label{defn.gamma.theta}
Suppose $\Gamma$ is a typing context and $\theta$ is an unknowns substitution.
Write $\Gamma\cent\theta$ when 
if $X\in\dom(\theta)$ then $X{:}\Box\tya\in\Gamma$ for some $\tya$ and $\Gamma\cent\theta(X):\Box\tya$.
\end{defn}

Proposition~\ref{prop.rtheta} is needed for Theorem~\ref{thrm.type.soundness} (soundness of the denotation).
It is slightly unusual that soundness of typing under substitution should be needed for soundness under taking denotations.
But the syntax is going to be \emph{part of} the denotational semantics---that is its point---and so substitution is part of how this denotation is calculated (see the case of $\Box r$ in Figure~\ref{fig.denot.terms}).
\begin{prop}
\label{prop.rtheta}
Suppose $\Gamma$ is a typing context and $\theta$ is an unknowns substitution and suppose $\Gamma\cent\theta$ (Definition~\ref{defn.gamma.theta}).
Then $\Gamma\cent r:\tya$ implies $\Gamma\cent r\theta:\tya$. 
\end{prop}
\begin{proof}
By a routine induction on the typing of $r$.
We consider four cases:
\begin{itemize*}
\item
\emph{The case of \rulefont{\Box I}.}\quad
Suppose $\Gamma\cent r:\tya$ and $\fa(r)=\varnothing$ so that $\Gamma\cent \Box r:\Box \tya$ by \rulefont{\Box I}.
By inductive hypothesis $\Gamma\cent r\theta:\tya$.
By Lemma~\ref{lemm.fa.rtheta} also $\fa(r\theta)=\varnothing$. 
We use \rulefont{\Box I} and the fact that $(\Box r)\theta=\Box(r\theta)$, and Proposition~\ref{prop.ws}.
\item
\emph{The case of \rulefont{Ext} for $X\in\dom(\theta)$.}\quad
By assumption in Definition~\ref{defn.substitution},\ $\theta(X)=\Box r'$ for some $r'$ with $\fa(r')=\varnothing$.
By assumption in Definition~\ref{defn.gamma.theta} $\varnothing\cent\theta(X):\Box\tya$.
By Definition~\ref{defn.sub} $(X_\at)\theta=r'$.
By Proposition~\ref{prop.ws} $\Gamma\cent r':\tya$ as required. 
\item
\emph{The case of \rulefont{{\fto}I}.}\quad
Suppose $\Gamma,a{:}\tya\cent r:\tyb$ so that by \rulefont{{\fto}I} $\Gamma\cent \lam{a{:}\tya}r:\tya\fto\tyb$.
By inductive hypothesis $\Gamma,a{:}\tya\cent r\theta:\tyb$.
We use \rulefont{{\fto}I}.
\item
\emph{The case of \rulefont{\Box E}.}\quad
Suppose $\Gamma,X{:}\Box\tya\cent r:\tyb$ and $\Gamma\cent s:\Box\tya$ so that by \rulefont{\Box E} $\Gamma\cent\letbox{X}{s}{r}:\tyb$.
Renaming if necessary, suppose $X\not\in\dom(\theta)$.
By inductive hypothesis $\Gamma,X{:}\Box\tya\cent r\theta:\tyb$ and $\Gamma\cent s\theta:\Box\tya$.
We use \rulefont{\Box E} and the fact that $(\letbox{X}{s}{r})\theta=\letbox{X}{s\theta}{r\theta}$.
\end{itemize*}
\end{proof}

\section{Denotational semantics for types and terms of the modal type system}
\label{sect.denotations}

We now develop a denotational semantics of the types and terms from Definitions~\ref{defn.types} and~\ref{defn.terms}.
The main definitions are in Figures~\ref{fig.denot.types} and~\ref{fig.denot.terms}.
The design is subtle, so there follows an extended discussion of the definition.

\subsection{Denotation of types}
\label{subsect.denot.types}

\begin{frametxt}
\begin{defn}
\label{defn.interpret.types}
Define $\hdenot{}{\tya}$ the \deffont{interpretation} of types by induction in Figure~\ref{fig.denot.types}. 
\end{defn}
\end{frametxt}

\begin{figure}[t]
$$
\begin{array}{r@{\ }l@{\quad}l}
\hdenot{}{o}=&\{\top^\hden,\bot^\hden\}
& \text{\it truth-values}
\\
\hdenot{}{\nat}=&\{0,1,2,\dots\}
& \text{\it natural numbers}
\\
\hdenot{}{\tya\fto\tyb}=&\hdenot{}{\tyb}^{\hdenot{}{\tya}}
&
\text{\it function-spaces}
\\
\hdenot{}{\Box\tya}=&\{\Box r\mid \varnothing\cent \Box r:\Box\tya\}\times\hdenot{}{\tya}
&
\text{\it closed syntax \& purported denotation}
\end{array}
$$
\caption{Denotational semantics of modal types}
\label{fig.denot.types}
\end{figure}

\begin{rmrk}
\label{rmrk.how.to.interpret.types}
$\hdenot{}{o}$ is a pair of truth-values, and $\hdenot{}{\mathbb N}$ is the set of natural numbers.
$\hdenot{}{\tyb}^{\hdenot{}{\tya}}$ is a function-space.\footnote{We could restrict this to computable functions or some other smaller set but we have our logician's hat on here, not our programmer's hat on: we \emph{want} the larger set.  This will make Corollary~\ref{corr.later} work.  If we chose a smaller, more sophisticated, and more complex notion of function-space here, then this would actually \emph{weaken} the results we then obtain from the semantics.}
No surprises here.

$z\in\hdenot{}{\Box\tya}$ is a pair $(\Box r,x)$.
We suggest the reader think of this as
\begin{itemize*}
\item
some syntax $\Box r$ and\footnote{We could drop the $\Box$ and just write $(r,x)$, but when we build the contextual system in Section~\ref{sect.contextual.system} the $\Box$ will fill with bindings (see Definition~\ref{defn.terms.c}) and cannot be dropped, so we keep it here.}
\item
its purported denotation $x$. 
\end{itemize*}
We say `purported' because there is no restriction that $x$ actually be a possible denotation of $r$. 
For instance, it is a fact that $\Box(0+1)::2\in\hdenot{}{\Box\nat}$, and $\Box(0+1)::2$ will not be the denotation of any $r$ such that $\varnothing\cent r:\nat$ (to check this, unpack Definition~\ref{defn.denot.terms} below).

So our semantics inflates: there are usually elements in $\hdenot{}{\Box\tya}$ that are not the denotation of any closed term.
The reader should remain calm; there are also usually elements in function-spaces that are not the denotation of any closed term.
The inflated elements in our semantics are an important part of our design.
\end{rmrk}

\begin{nttn}
\label{nttn.list.notation}
We will want to talk about nested pairs of the form $(x_1,(x_2,\dots,(x_n,x_{n+1})))$.
Accordingly we will use list notation, writing $x_1::x_2$ for $(x_1,x_2)$ and $x_1::\dots::x_n::x_{n+1}$ for $(x_1,(x_2,\dots,(x_n,x_{n+1})))$.
See for instance Remark~\ref{rmrk.next.remark}, Figure~\ref{fig.denot.terms}, and Subsection~\ref{subsect.den.letbox}.
\end{nttn}

\begin{rmrk}
\label{rmrk.next.remark}
Note that as standard, distinct syntax may have equal denotation.
For instance, $\Box(0+1)::1$ and $\Box(1+0)::1$ are not equal in $\hdenot{}{\nat}$.
\end{rmrk}

\begin{rmrk}
\label{rmrk.why.inflate}
Why do we inflate?
Surely it is both simpler and more intuitive to take 
$\hdenot{}{\Box\tya}$ to be $\{\Box r\mid \varnothing\cent \Box r:\Box\tya\}$.

We could do this, but then later on in Definition~\ref{defn.denot.terms} we would not be able to give a denotation to terms by induction on their syntax.

The problem is that our types, and terms, are designed to permit generation of syntax at modal type. 
Thus, our design brief is to allow dynamic (runtime) generation of syntax.
With the `intuitive' definition above, there is no guarantee of an inductively decreasing quantity; the runtime can generate syntax of any size.
To see this in detail, see Subsection~\ref{subsubsect.explain}.

The design of $\hdenot{}{\Box\tya}$ in Figure~\ref{fig.denot.types} gets around this by insisting, at the very moment we assert some denotation of a term $r$ of type $\Box\tya$---i.e. some syntax $r'$ of type $\tya$---to simultaneously volunteer a denotation for $r'$---i.e. an element in the denotation of $\tya$.
(As mentioned in Remark~\ref{rmrk.how.to.interpret.types} this denotation might be in some sense mistaken, but perhaps surprisingly that will not matter.)
\end{rmrk}

\subsection{Denotation of terms}

We now set about interpreting terms in the denotation for types from Definition~\ref{defn.interpret.types}.
The main definition is Definition~\ref{defn.denot.terms}.
First, however, we need:
\begin{itemize*}
\item
some tools to handle the `syntax and purported denotation' design of $\hdenot{}{\Box\tya}$ (Definition~\ref{defn.hd.tl}); and 
\item
a suitable notion of valuation (Definition~\ref{defn.valuation}).
\end{itemize*}
We then discuss the design of the definitions.

Recall from Notation~\ref{nttn.list.notation} that we may use list notation and write $\Box r::x$ for $(\Box r,x)$.
\begin{defn}
\label{defn.hd.tl}
We define $\hd$ and $\tl$ on $x\in\hdenot{}{\tya}$ (Definition~\ref{defn.interpret.types}) as follows:
\begin{frametxt}
\begin{itemize*}
\item
If $x\in\hdenot{}{o}$ or $\hdenot{}{\nat}$ or $\hdenot{}{\tya\fto\tyb}$ then $\hd(x)=x$ and $\tl(x)$ is undefined. 
\item
If $(\Box r,x)\in\hdenot{}{\Box\tya}$ then $\hd((\Box r,x))=\Box r$ (first projection) and $\tl((\Box r,x))=x$ (second projection).
\end{itemize*}
\end{frametxt}
\end{defn}

\begin{defn}
\label{defn.valuation}
A \deffont{valuation} $\varsigma$ is a finite partial function on $\atoms\cup\unknowns$.
Write $\varsigma[X{\ssm}{x}]$ for the valuation such that:
\begin{itemize*}
\item
$(\varsigma[X{\ssm}{x}])(X)={x}$.
\item
$(\varsigma[X{\ssm}{x}])(Y)=\varsigma(Y)$ if $\varsigma(Y)$ is defined, for all $Y$ other than $X$.
\item
$(\varsigma[X{\ssm}{x}])(a)=\varsigma(a)$ if $\varsigma(a)$ is defined.
\item
$(\varsigma[X{\ssm}{x}])$ is undefined otherwise.
\end{itemize*}
Define $\varsigma[a{\ssm}{x}]$ similarly.
\end{defn}

\begin{defn}
\label{defn.gamma.varsigma}
Suppose $\Gamma$ is a typing context and $\varsigma$ a valuation.
Write $\Gamma\cent\varsigma$ when:
\begin{frametxt}
\begin{enumerate*}
\item
$\dom(\Gamma)=\dom(\varsigma)$.
\item
If $a\in\dom(\varsigma)$ then $a{:}\tya\in\Gamma$ for some $\tya$ and $\varsigma(a)\in\hdenot{}{\tya}$.
\item
If $X\in\dom(\varsigma)$ then $X{:}\Box\tya\in\Gamma$ for some $\tya$ and $\varsigma(X)\in\hdenot{}{\Box\tya}$.
\end{enumerate*}
\end{frametxt}
\end{defn}

\begin{rmrk}
\label{rmrk.box.values}
Unpacking Definition~\ref{defn.interpret.types}, clause~3 of Definition~\ref{defn.gamma.varsigma} (the one for $X$) means that $\varsigma(X)=\Box r'::x$ where $\varnothing\cent\Box r':\Box\tya$ and $x\in\hdenot{}{\tya}$.
Note also that by the form of the derivation rules in Figure~\ref{fig.modal.types}, it follows that $\varnothing\cent r':\tya$.
So an intuition for $\varsigma(X)$ (cf. Remark~\ref{rmrk.how.to.interpret.types}) is this---
\begin{quote}
``$\varsigma(X)$ is some closed syntax $r'$ (presented as $\Box r'\in\hdenot{}{\Box\tya}$), and a candidate denotation for it $x\in\hdenot{}{\tya}$'',
\end{quote}
---or more concisely this:
\begin{quote}
``$\varsigma(X)$ is a pair of syntax and denotation.''
\end{quote}
\end{rmrk}

\begin{defn}
\label{defn.varsigmaunknowns}
Write $\varsigma_\unknowns$ for the unknowns substitution (Definition~\ref{defn.substitution}) such that
\begin{frameqn}
\varsigma_\unknowns(X)=\hd(\varsigma(X))
\end{frameqn}
if $\varsigma(X)$ is defined, and $\varsigma_\unknowns$ is undefined otherwise.
\end{defn}

\begin{defn}
\label{defn.denot.terms}
For each constant $C:\tya$ other than $\top$, $\bot$, and $\tf{isapp}$ fix some interpretation $C^\hden$ which is an element $C^\hden\in\hdenot{}{\tya}$.
Suppose $\Gamma\cent\varsigma$ and $\Gamma\cent r:\tya$.
\begin{frametxt}
An \deffont{interpretation} of terms $\hdenot{\varsigma}{r}$ is defined in Figure~\ref{fig.denot.terms}.
\end{frametxt}
\end{defn}

\begin{figure}[t]
$\varsigma[a{\ssm}x]$ and $\varsigma[X{\ssm}x]$ from Definition~\ref{defn.valuation}.
$\varsigma_\unknowns$ from Definition~\ref{defn.varsigmaunknowns}.
$$
\begin{array}{r@{\ }l@{\qquad}l}
\hdenot{\varsigma}{\top}=&\top^\hden
\\
\hdenot{\varsigma}{\bot}=&\bot^\hden
\\
\hdenot{\varsigma}{a}=&\varsigma(a) &(a\in\dom(\varsigma))
\\
\hdenot{\varsigma}{\lam{a{:}\tya}r}=&(x{\in}\hdenot{}{\tya}\mapsto \hdenot{\varsigma[a{\ssm}x]}{r})
\\
\hdenot{\varsigma}{r'r}=&\hdenot{\varsigma}{r'}\,\hdenot{\varsigma}{r}
\\
\hdenot{\varsigma}{\Box r}=&(\Box(r\varsigma_\unknowns))::\hdenot{\varsigma}{r}
\\
\hdenot{\varsigma}{X_\at}=&\tl(\varsigma(X))
\\
\hdenot{\varsigma}{\letbox{X}{s}{r}}=&\hdenot{\varsigma[X{\ssm}\hdenot{\varsigma}{s}]}{r}
\\
\hdenot{\varsigma}{\tf{isapp}_\tya}(\Box(r'r''))=& \top^\hden 
\\
\hdenot{\varsigma}{\tf{isapp}_\tya}(\Box(r))=&\bot^\hden &(\Forall{r',r''}r\neq r'r'') 
\end{array}
$$
\caption{Denotational semantics of terms of the modal type system}
\label{fig.denot.terms}
\end{figure}

In Subsection~\ref{subsect.denot.discuss} we discuss the design of $\hdenot{\varsigma}{r}$, with examples.
In Subsection~\ref{subsect.denot.results} we prove some results about it.

\subsection{Discussion of the denotation}
\label{subsect.denot.discuss}

\subsubsection{About the term-formers}

The denotations of $\top$ and $\bot$ are as expected.
To give a denotation to an atom $a$, we just look it up using $\varsigma$, also as expected.
The definitions of $\lam{a{:}\tya}r$ and $r'r$ are also as standard.

As promised in Subsection~\ref{subsect.denot.types},\ $\hdenot{\varsigma}{\Box r}$ returns a pair of a syntax and its denotation.

$\tf{isapp}_\tya$ is there to illustrate concretely how we can express programming on syntax of box types: it takes a syntax argument and checks whether it is a syntactic application.%
\footnote{We know non-trivial pattern-matching on applications exists in our meta-logic because our meta-logic is English; $\hdenot{\varsigma}{\tf{isapp}}$ is a function on a set of syntax and we can define whatever operation we can define, on that set.} 
Of course many other such functions are possible, and if we want them we can add them as further constants (just as we might add $+$, $*$, and/or recursion as constants, given a type for numbers).

\subsubsection{Example: denotation of $\letbox{X}{\Box(1+2)}{\Box\Box X_\at}$}
\label{subsect.den.letbox}

To illustrate how Figure~\ref{fig.denot.terms} works, we calculate the denotation of $\letbox{X}{\Box(1+2)}{\Box\Box X_\at}$.
We reason as follows, where for compactness and clarity we write $\varsigma$ for the valuation $[X{\ssm}\Box(1{+}2)::3]$:
$$
\hspace{-2ex}\begin{array}{r@{\ }l@{\qquad}l}
\hdenot{\varnothing}{\letbox{X}{\Box(1{+}2)}{\Box\Box X}}=&\hdenot{[X{\ssm}\hdenot{\varnothing}{\Box(1{+}2)}]}{\Box\Box X_\at}
\\
=&\hdenot{\varsigma}{\Box\Box X_\at}
\\
=&\Box((\Box X_\at)[X{\ssm}\Box(1{+}2)]) :: \hdenot{\varsigma}{\Box X_\at}
\\
=&\Box\Box(1{+}2) :: \hdenot{\varsigma}{\Box X_\at}
\\
=&\Box\Box(1{+}2) :: \Box(X_\at[X{\ssm}\Box(1{+}2)]) :: \hdenot{\varsigma}{X_\at}
\\
=&\Box\Box(1{+}2) :: \Box(1{+}2) :: \hdenot{\varsigma}{X_\at}
\\
=&\Box\Box(1{+}2) :: \Box(1{+}2) :: \tl(\Box(1{+}2)::3)
\\
=&\Box\Box(1{+}2) :: \Box(1{+}2) :: 3 
\end{array}
$$
We leave it to the reader to verify that $\hdenot{\varnothing}{\Box(1{+}2)}=\Box(1{+}2)::3$ and that $X_\at[X{\ssm}\Box(1{+}2)]=1{+}2$.

Note that `$1+2$' and `$\Box(1+2)$' are different; $1+2$ denotes $3$ whereas $\Box(1+2)$ denotes the pair `The syntax $1+2$, with associated extension $3$'.
In some very special cases where the set of possible denotations is rather small (finite or countable), the distinction between terms and their denotations can be hard to see, though it is still there.
Usually sets of denotations are `quite large' and sets of syntax are `quite small', but sometimes this relationship is reversed: there are `somewhat more' terms denoting numbers, than numbers\footnote{$6+5$ and $5+6$ denote the same number, whose calculation we leave as an exercise to the energetic reader.} (but much fewer terms denoting functions from numbers to numbers than functions from numbers to numbers).
See Corollary~\ref{corr.later} and Example~\ref{xmpl.unevaluate.nat}.

Note also the difference between the valuation $\varsigma=[X{\ssm}\Box(1{+}2)::3]$ and the substitution $[X{\ssm}\Box(1{+}2)]$.
The first is a valuation because it maps $X$ to $\hdenot{}{\Box\nat}$, the second is a substitution because it makes $X$ to a term of type $\Box\nat$. 

Sometimes a mapping can be both valuation and substitution; for instance $[a{\ssm}3]$ is a valuation ($a$ maps to an element of $\hdenot{}{\nat}$), and is also a substitution.

\subsubsection{Why the natural version does not work}
\label{subsubsect.explain}

Natural versions of Definitions~\ref{defn.interpret.types} and~\ref{defn.denot.terms} take
\begin{itemize*}
\item
the denotation of box type to be just boxed syntax rather than a pair of boxed syntax and denotation $\hdenot{}{\Box\tya}=\{\Box r\mid \varnothing\cent\Box r:\Box\tya\}$, and
\item 
$\hdenot{\varsigma}{\Box r}=\Box(r\varsigma_\unknowns)$ and
\item
$\hdenot{\varsigma}{X_\at}=\hdenot{\varnothing}{r}$ where $\varsigma(X)=\Box r$.
\end{itemize*}
However, this seems not to work; $\varsigma(X)$ need not necessarily be a smaller term than $X$ so the `definition' above is not inductive.
This is not just a hypothetical issue: a term of the form $\hdenot{\varsigma}{\letbox{X}{s}{r}}$ may cause $\varsigma(X)$ to be equal to $\hdenot{\varsigma}{s}$, and $s$ might generate syntax of any size.

The previous paragraph is not a mathematical proof; aside from anything else we have left the notion `size of term' unspecified.
The reader can experiment with different candidates: obvious `subterm of', `depth of', and `number of symbols' of are all vulnerable to the problem described above, as is a more sophisticated notion of size which gives $X$ size $\omega$ the least infinite cardinal---since we can generate multiple copies of terms of the form $\letbox{X}{r}{s}$, and even if this is closed it can contain bound copies of $X$.

\subsubsection{Example: denotation of $\f{exp}\,2$}
\label{subsubsect.exp.denot}

Recalling Subsection~\ref{subsubsect.exp}, we calculate the denotation of $\hdenot{\varnothing}{\f{exp}\,2}$ where $\f{exp}$ is specified by: 
$$
\begin{array}[t]{r@{\ }l}
\f{exp}\ 0\Rightarrow&\Box\lam{b{:}\nat}1 
\\ 
\f{exp}\ (\f{succ}(n))\Rightarrow&\letbox{X}{\f{exp}\,n}{\Box (\lam{b{:}\nat}b*(X_\at b) )} . 
\end{array}
$$
We sketch part of the calculation:
$$
\begin{array}{r@{\ }l}
\hdenot{\varnothing}{\f{exp}\,(\f{succ}\,(\f{succ}\, 0))}=&\hdenot{\varnothing}{\letbox{X}{\f{exp}\,(\f{succ}\, 0)}{\Box(\lam{b{:}\nat}b*(X_\at b)})}
\\[1ex]
=&\hdenot{[X{\ssm}\hdenot{\varnothing}{\f{exp}\,(\f{succ}\,0)}]}{\Box(\lam{b{:}\nat}b*(X_\at b))}
\\[1.5ex]
=&\Box(\lam{b{:}\nat}b*(X_\at b))[X{\ssm}\hd\hdenot{\varnothing}{\f{exp}\,(\f{succ}\,0))}]
\\
&\quad ::\hdenot{[X{\ssm}\hdenot{\varnothing}{\f{exp}\,(\f{succ}\,0)}]}{\lam{b{:}\nat}b*(X_\at b)}
\\
\raisebox{3pt}{$\cdots$}&
\\
=&\Box(\lam{b{:}\nat}b*(\lam{b{:}\nat}b*((\lam{b{:}\nat}1)b)b))::(x\in\nat{\mapsto}x*x)
\end{array}
$$

\subsubsection{Example: denotation of terms for axioms \rulefont{T} and \rulefont{4}}
\label{subsubsect.4.denot}

In Subsection~\ref{subsubsection.one.line} we considered the terms
$$
\begin{array}{l}
\lam{a{:}\Box\tya}\letbox{X}{a}{X_\at}\ :\ \Box\tya\fto\tya\quad\text{and}
\\
\lam{a{:}\Box\tya}\letbox{X}{a}{\Box\Box X_\at}\ :\ \Box\tya\fto\Box\Box\tya 
\end{array}
$$
which implement the modal logic axioms \rulefont{T} and \rulefont{4}.
We now describe their denotations, without working:
\begin{itemize*}
\item
$\hdenot{\varnothing}{\lam{a{:}\Box\tya}\letbox{X}{a}{X_\at}}$ maps $\Box r::tl\in\hdenot{}{\Box\tya}$ to $tl$.
\item
$\hdenot{\varnothing}{\lam{a{:}\Box\tya}\letbox{X}{a}{\Box\Box X_\at}}$ maps $\Box r::tl\in\hdenot{}{\Box\tya}$ to $\Box\Box r::\Box r::tl$.
\end{itemize*}

\subsection{Results about the denotation}
\label{subsect.denot.results}

We need a technical result and some notation for Proposition~\ref{prop.typesub}:
\begin{lemm}
\label{lemm.varsigma.to.varsigmaun}
If $\Gamma\cent\varsigma$ (Definition~\ref{defn.gamma.varsigma})
then $\Gamma\cent\varsigma_\unknowns$ (Definition~\ref{defn.gamma.theta}).
\end{lemm}
\begin{proof}
If $X\not\in\dom(\varsigma)$ then $X\not\in\dom(\varsigma_\unknowns)$.

Suppose $X\in\dom(\varsigma)$.
By Definition~\ref{defn.varsigmaunknowns} $\varsigma_\unknowns(X)=\hd(\varsigma(X))$. 
By Definition~\ref{defn.gamma.varsigma} $\varsigma_\unknowns(X)\in\hdenot{}{\Box\tya}$ for some $\tya$.
Unpacking Figure~\ref{fig.denot.types} this implies that $\varsigma_\unknowns(X)=\Box r$ for some $\varnothing\cent r:\tya$, and we are done.
\end{proof}

Proposition~\ref{prop.typesub} relies on a dual role played by syntax in $\varsigma_\unknowns$.
It is coerced between denotation and syntax in $\rulefont{\Box I}$, and `in the other direction' in \rulefont{Ext}. 
Proposition~\ref{prop.typesub} expresses this important dynamic in the mathematics of the paper. 
Technically, the result is needed for the case of \rulefont{\Box I} in the proof of Theorem~\ref{thrm.type.soundness}.
Recall the notation $\Gamma|_U$ from Notation~\ref{nttn.restrict}.
\begin{prop}
\label{prop.typesub}
Suppose $\Gamma\cent r:\tya$ and $\Gamma\cent \varsigma$.
Then $\Gamma|_\atoms\cent r\,\varsigma_\unknowns:\tya$. 

\noindent ($\varsigma_\unknowns$ is defined in Definition~\ref{defn.varsigmaunknowns}; its action on $r$ is defined in Definition~\ref{defn.sub}.)
\end{prop}
\begin{proof}
By Lemma~\ref{lemm.varsigma.to.varsigmaun} $\Gamma\cent\varsigma_\unknowns$.
By Proposition~\ref{prop.rtheta} $\Gamma\cent r\,\varsigma_\unknowns:\tya$.
By Lemma~\ref{lemm.fa.rtheta} $\fa(r\varsigma_\unknowns)=\fa(r)$.
Now it is a fact that $\fa(r)\subseteq\dom(\Gamma|_\atoms)$, so by Proposition~\ref{prop.ws} $\Gamma|_\atoms\cent r\,\varsigma_\unknowns:\tya$ as required.
\end{proof}

\begin{frametxt}
\begin{thrm}[Soundness]
\label{thrm.type.soundness}
If $\Gamma\cent r:\tya$ and $\Gamma\cent\varsigma$ then $\hdenot{\varsigma}{r}$ is defined and $\hdenot{\varsigma}{r}\in\hdenot{}{\tya}$.
\end{thrm}
\end{frametxt}
\begin{proof}
By induction on the derivation of $\Gamma\cent r:\tya$.
Most of the rules follow by properties of sets and functions.
We consider the interesting cases:
\begin{itemize*}
\item
\emph{Rule \rulefont{\Box I}.}\quad
Suppose $\Gamma\cent r:\tya$ and $\fa(r){=}\varnothing$ so that by \rulefont{\Box I} $\Gamma\cent\Box r:\tya$.

Suppose $\Gamma\cent\varsigma$.
Then by inductive hypothesis $\hdenot{\varsigma}{r}\in\hdenot{}{\tya}$.
Also, by Proposition~\ref{prop.typesub} $\varnothing\cent r\varsigma_\unknowns:\tya$ 

It follows by Definition~\ref{defn.interpret.types} that
$$
\hdenot{\varsigma}{\Box r}=(\Box(r\varsigma_\unknowns))::\hdenot{\varsigma}{r}\in\hdenot{}{\Box\tya}
$$ 
as required.
\item
\emph{Rule \rulefont{\Box E}.}\quad
Suppose $\Gamma,X{:}\Box\tya\cent r:\tyb$ and $\Gamma\cent s:\Box\tya$ so that by \rulefont{\Box E} $\Gamma\cent\letbox{X}{s}{r}:\tya$.

Suppose $\Gamma\cent\varsigma$.
By inductive hypothesis for $\Gamma\cent s:\Box\tya$ we have $\hdenot{\varsigma}{s}\in\hdenot{}{\Box\tya}$ and so there is some term $s'$ and some $x\in\hdenot{}{\tya}$ such that $(\Box s')::x=\hdenot{\varsigma}{s}$ and $\varnothing\cent \Box s':\Box\tya$.
Unpacking Definition~\ref{defn.gamma.varsigma}, $\Gamma,X{:}\Box\tya\cent\varsigma[X{\ssm}(\Box s')::x]$.
By inductive hypothesis for $\Gamma,X{:}\Box\tya\cent r:\tyb$ we have 
$$
\hdenot{\varsigma[X{\ssm}(\Box s')::x]}{r}\in\hdenot{}{\tyb}
$$
and using Definition~\ref{defn.denot.terms} we have 
$$
\hdenot{\varsigma}{\letbox{X}{s}{r}}=\hdenot{\varsigma[X{\ssm}(\Box s')::x]}{r}\in\hdenot{}{\tyb}
$$
as required.
\item
\emph{Rule \rulefont{Ext}.}\quad
By \rulefont{Ext} $\Gamma,X{:}\Box\tya\cent X_\at{:}\Box\tya$.

Suppose $\Gamma,X{:}\Box\tya\cent\varsigma$.
Unpacking Definition~\ref{defn.gamma.varsigma}, this means that $\varsigma(X)=(\Box s')::x$ for some $s'$ and $x$ such that $\varnothing\cent\Box s':\Box\tya$ and $x\in\hdenot{}{\tya}$.
From Definition~\ref{defn.denot.terms} $\hdenot{\varsigma}{X_\at}=x\in\hdenot{}{\tya}$ as required.
\item
\emph{Rule \rulefont{Hyp}.}\quad
Suppose $\Gamma,a{:}\tya\cent\varsigma$.
By Definition~\ref{defn.gamma.varsigma} this means that $\varsigma(a)\in\hdenot{}{\tya}$.
By Definition~\ref{defn.denot.terms} $\hdenot{\varsigma}{a}=\varsigma(a)$.
The result follows.
\end{itemize*}
\end{proof}

\begin{frametxt}
\begin{corr}
\label{corr.later}
There is no term $s$ such that 
$\varnothing\cent s:(\nat\fto\nat)\fto\Box(\nat\fto\nat)$
is typable and such that the map $\lam{x{\in}\nat^\nat}\hd(\hdenot{\varnothing}{s}\,x)\in \hd(\hdenot{}{\Box(\nat\fto\nat)})^{\hdenot{}{\nat\fto\nat}}$ is injective.
\end{corr}
\end{frametxt}
\begin{proof}
$\hdenot{}{\nat\fto\nat}$ is an uncountable set whereas $\hd(\hdenot{}{\Box(\nat\fto\nat)})=\{r\mid \varnothing\cent r:\nat\fto\nat\}$ is countable.
The result follows from Theorem~\ref{thrm.type.soundness}. 
\end{proof}

\begin{xmpl}
\label{xmpl.unevaluate.nat}
By Corollary~\ref{corr.later} there can be no term representing a function 
which reifies an element of $\hdenot{}{\tya}$ to corresponding syntax.

Of course, there might be a term which reifies those elements of $\hdenot{}{\tya}$ that are representable by syntax.
For specific `sufficiently small' $\tya$, this might even include all of $\hdenot{}{\tya}$.

For example, if $\tya = \nat$ then the following function does the job:
$$
\begin{array}[t]{r@{\ }l}
\f{reifyNat}\ 0\Rightarrow&\Box 0
\\ 
\f{reifyNat}\ (\tf{succ}(n))\Rightarrow&\letbox{X}{\f{reifyNat}(n)}{\Box (X_\at{+}1)}.
\end{array}
$$
\end{xmpl}

\begin{rmrk}
Similar arguments to those used in Corollary~\ref{corr.later} and Example~\ref{xmpl.unevaluate.nat} also justify why the Haskell programming language has a \emph{Show} function for certain types, but not for function types.\footnote{See \url{haskell.org/haskellwiki/Show_instance_for_functions}, retrieved on January 20, 2012.} 
We chose full function spaces in Figure~\ref{fig.denot.terms}, so that the models for which we prove soundness in Theorem~\ref{thrm.type.soundness} would be large, and we did that so that the proof of Corollary~\ref{corr.later} would become relatively easy. 
Careful consideration has gone into the precise designs of $\hdenot{}{\tyb}^\hdenot{}{\tya}$ and $\hdenot{}{\Box\tya}$.

We will later on in Corollary~\ref{corr.later.c} prove a similar result for the contextual system, and then later still in Corollary~\ref{corr.more.c} surprisingly leverage this to a result which even works for functions to all of $\hdenot{}{\Box(\nat\fto\nat)}$ rather than just to the (much smaller) $\hd(\hdenot{}{\Box(\nat\fto\nat)})$.
\end{rmrk}

\section{Reduction}
\label{sect.reduction}

We have Theorem~\ref{thrm.type.soundness} (soundness) and Corollary~\ref{corr.later} (impossibility in general of reifying denotation to syntax).
The other major property of interest is that typing and denotation are consistent with a natural notion of reduction on terms.

So we now turn our attention to the lemmas leading up to Proposition~\ref{prop.type.soundness} and Theorem~\ref{thrm.red.sound}.

\subsection{Results concerning substitution on atoms}

Recall from Definition~\ref{defn.sub} the definition of the atoms-substitution action.
Lemma~\ref{lemm.typesub.a} is a counterpart to Proposition~\ref{prop.typesub}.
We had to prove Proposition~\ref{prop.typesub} earlier because calculating the denotation $\hdenot{\varsigma}{\Box r}$ in Figure~\ref{fig.denot.terms} involves calculating $r\varsigma_\unknowns$ (an unknowns-substitution applied to a term).\footnote{In the contextual system, calculating the denotation will involve atoms-substitution as well.}
Now we are working towards reduction, and $\beta$-reduction can generate atoms-substitution, so we need Lemma~\ref{lemm.typesub.a}.
\begin{lemm}
\label{lemm.typesub.a}
Suppose $\Gamma,a{:}\tyb\cent r:\tya$ and $\Gamma\cent s:\tyb$. 
Then $\Gamma\cent r[a{\ssm}s]:\tya$.
\end{lemm}
\begin{proof}
By a routine induction on the typing of $r$.
We consider three cases:
\begin{itemize*}
\item
\emph{The case of \rulefont{\Box I}.}\quad
Suppose $\Gamma,a{:}\tyb\cent r:\tya$ and $\fa(r){=}\varnothing$ so that $\Gamma,a{:}\tyb\cent \Box r:\Box \tya$ by \rulefont{\Box I}.
But then by Lemma~\ref{lemm.standard.fv.sub} $r[a{\ssm}s]=r$, and the result follows from Proposition~\ref{prop.ws}.
\item
\emph{The case of \rulefont{Ext}}\quad is similar to that of \rulefont{\Box I}.
\item
\emph{The case of \rulefont{\Box E}.}\quad
Using the fact from Definition~\ref{defn.sub} that 
$$
(\letbox{X}{s'}{r})[a{\ssm}s]\ =\ \letbox{X}{s'[a{\ssm}s]}{r[a{\ssm}s]}.
$$
\end{itemize*}
\end{proof}

\begin{lemm}
\label{lemm.sub.var}
Suppose $\Gamma,a{:}\tyb\cent r:\tya$ and $\Gamma\cent s:\tyb$, and suppose $\Gamma\cent\varsigma$.
Then $\hdenot{\varsigma}{r[a{\ssm}s]}=\hdenot{\varsigma[a{\ssm}\hdenot{\varsigma}{s}]}{r}$.
\end{lemm} 
\begin{proof}
By a routine induction on the derivation of $\Gamma,a{:}\tyb\cent r:\tya$ (Figure~\ref{fig.modal.types}).
We consider three cases:
\begin{itemize*}
\item
\emph{The case of \rulefont{\Box I}.}\quad
We use Lemma~\ref{lemm.standard.fv.sub} and Proposition~\ref{prop.ws} (as in the case of \rulefont{\Box I} in the proof of Lemma~\ref{lemm.typesub.a}). 
\item
\emph{The case of \rulefont{Ext}.}\quad
By \rulefont{Ext} $\Gamma,a{:}\tyb,X{:}\tya\cent X_\at:\tya$.
By definition $X_\at[a{\ssm}s]=X_\at$.
We use Proposition~\ref{prop.ws}.
\item
\emph{The case of \rulefont{Hyp} for $a$.}\quad
By \rulefont{Hyp} $\Gamma,a{:}\tyb\cent a:\tyb$.
By assumption $\Gamma,a{:}\tyb\cent \varsigma$ so unpacking Definition~\ref{defn.gamma.varsigma}, $\varsigma(a)\in\hdenot{}{\tyb}$.
By Figure~\ref{fig.denot.terms} $\varsigma(a)=\hdenot{\varsigma}{a}$, and we are done.
\end{itemize*}
\end{proof}

Proposition~\ref{prop.ws.denote} can be viewed as a denotational counterpart of Proposition~\ref{prop.ws}:
\begin{prop}
\label{prop.ws.denote}
Suppose $\Gamma\cent r:\tya$ and $\Gamma\cent\varsigma$ and $\Gamma\cent\varsigma'$. 
Suppose 
$\varsigma(a)=\varsigma'(a)$ for every $a\in\fa(r)$ and 
$\varsigma(X)=\varsigma'(X)$ for every $X\in\fa(r)$.

Then $\hdenot{\varsigma}{r}=\hdenot{\varsigma'}{r}$.
\end{prop}
\begin{proof}
By a routine induction on $r$.
\end{proof}

\begin{lemm}
\label{lemm.suba.to.varsigma}
Suppose $\Gamma,a{:}\tya\cent r:\tyb$ and $\Gamma\cent s:\tya$,\ and suppose $\Gamma\cent\varsigma$.
Then 
$$
\denot{}{\varsigma}{(\lam{a{:}\tya}r)s}=\denot{}{\varsigma[a{\ssm}\denot{}{\varsigma}{s}]}{r}.
$$
\end{lemm}
\begin{proof}
We unpack the cases of $\lambda$ and application in Definition~\ref{defn.denot.terms}.
\end{proof}

\subsection{Results concerning substitution on unknowns}

\begin{lemm}
\label{lemm.letbox.denot}
Suppose $\Gamma\cent(\letbox{X}{s}{r}):\tya$ and $\Gamma\cent\varsigma$.
Then 
$$
\hdenot{\varsigma}{\letbox{X}{s}{r}}=\hdenot{\varsigma[X{\ssm}\hdenot{\varsigma}{s}]}{r} .
$$
\end{lemm}
\begin{proof}
We just unpack the clause for $\letbox{X}{s}{r}$ in Figure~\ref{fig.denot.terms} (well-definedness is from Theorem~\ref{thrm.type.soundness}).
\end{proof}

\begin{lemm}
\label{lemm.comm.theta.X}
Suppose $\theta$ is an unknowns-substitution (Definition~\ref{defn.substitution}).
Suppose $X\not\in\dom(\theta)$ and suppose $\fv(\theta(Z))=\varnothing$ for every $Z\in\dom(\theta)$.

Then $r[X{\ssm}\Box s]\theta=r\theta[X{\ssm}\Box(s\theta)]$.
\end{lemm}
\begin{proof}
By a routine induction on $r$.
The interesting case is $X_\at$, for which it is easy to check that:
$$
X_\at\theta[X{\ssm}(\Box s)\theta]=s\theta
\quad\text{and}\quad
X_\at[X{\ssm}\Box s]\theta=s\theta.
$$
\end{proof}

\maketab{tab1}{@{\hspace{-2em}}R{8em}@{\ }L{17em}@{}L{10em}}

\begin{lemm}
\label{lemm.Xsub.denot}
Suppose $\Gamma,X{:}\Box\tyb\cent r:\tya$ and $\Gamma\cent \Box s:\Box\tyb$, and suppose $\Gamma\cent\varsigma$.
Then $\hdenot{\varsigma}{r[X{\ssm}\Box s]}=\hdenot{\varsigma[X{\ssm}\hdenot{\varsigma}{\Box s}]}{r}$. 
\end{lemm}
\begin{proof}
By induction on the derivation of $\Gamma,X{:}\Box\tyb\cent r:\tya$.
\begin{itemize*}
\item
\emph{The case of \rulefont{\Box I}.}\quad
Suppose $\Gamma,X{:}\Box\tyb\cent r:\tya$ and $\fa(r)=\varnothing$ so that by \rulefont{\Box I} $\Gamma,X{:}\Box\tyb\cent\Box r:\Box\tya$. 
We sketch the necessary reasoning:
\begin{tab1}
\hdenot{\varsigma}{(\Box r)[X{\ssm}\Box s]}=&\hdenot{\varsigma}{\Box(r[X{\ssm}\Box s])}
&\text{Definition~\ref{defn.sub}}
\\
=&\Box(r[X{\ssm}\Box s])\varsigma_\unknowns::\hdenot{\varsigma}{r[X{\ssm}\Box s]}
&\text{Figure~\ref{fig.denot.terms}}
\\
=&\Box(r[X{\ssm}\Box s])\varsigma_\unknowns::\hdenot{\varsigma[X{\ssm}\hdenot{\varsigma}{\Box s}]}{r}
&\text{Ind. Hyp.}
\\
=&\Box(r\varsigma_\unknowns[X{\ssm}\Box(s\varsigma_\unknowns)])::\hdenot{\varsigma[X{\ssm}\hdenot{\varsigma}{\Box s}]}{r}
&\text{Lemma~\ref{lemm.comm.theta.X}}
\\[1.5ex]
\hdenot{\varsigma[X{\ssm}\hdenot{\varsigma}{\Box s}]}{\Box r}
=&(\Box r)(\varsigma[X{\ssm}\hdenot{\varsigma}{\Box s}])_\unknowns::\hdenot{\varsigma[X{\ssm}\hdenot{\varsigma}{\Box s}]}{r}
&\text{Figure~\ref{fig.denot.terms}}
\\
=&(\Box r)\varsigma_\unknowns[X{\ssm}\Box(s\varsigma_\unknowns)]::\hdenot{\varsigma[X{\ssm}\hdenot{\varsigma}{\Box s}]}{r}
&\text{Figure~\ref{fig.denot.terms}}
\\
=&\Box(r\varsigma_\unknowns[X{\ssm}\Box(s\varsigma_\unknowns)])::\hdenot{\varsigma[X{\ssm}\hdenot{\varsigma}{\Box s}]}{r}
&\text{Definition~\ref{defn.sub}}
\end{tab1}
\item
\emph{The case of \rulefont{Ext} for $X$.}\quad
By \rulefont{Ext} $\Gamma,X{:}\Box\tyb\cent X_\at:\tyb$.
Then we reason as follows:
\begin{tab1}
\hdenot{\varsigma[X{\ssm}\hdenot{\varsigma}{\Box s}]}{X_\at}=&\tl(\hdenot{\varsigma}{\Box s})
&\text{Figure~\ref{fig.denot.terms}}
\\
=&\hdenot{\varsigma}{s}
&\text{Figure~\ref{fig.denot.terms}}
\\[1.5ex]
\hdenot{\varsigma}{X_\at[X{\ssm}\Box s]}=&\hdenot{\varsigma}{s}
&\text{Definition~\ref{defn.sub}}
\end{tab1}
\end{itemize*}
\end{proof}

\subsection{Reduction}
\label{subsect.reduction}

\begin{frametxt}
\begin{defn}
\label{defn.reduction}
Define \deffont{$\beta$-reduction} $r\bto r'$ inductively by the rules in Figure~\ref{fig.reduction}.
\end{defn}
\end{frametxt}

\begin{figure}[t]
$r[a{\ssm}s]$ and $r[X{\ssm}s]$ from Definition~\ref{defn.sub}.
$$
\begin{array}{c@{\qquad}c}
\begin{prooftree}
\phantom{h}
\justifies
(\lam{a{:}\tya}r)r'\bto r[a{\ssm}r']
\using\rulefont{\beta}
\end{prooftree}
&
\begin{prooftree}
\phantom{h}
\justifies
\letbox{X}{\Box s}{r}\bto r[X{\ssm}\Box s]
\using\rulefont{\beta_\Box}
\end{prooftree}
\\[5ex]
\begin{prooftree}
r\bto r'\quad s\bto s'
\justifies
rs\bto r's'
\using\rulefont{cnga}
\end{prooftree}
&
\begin{prooftree}
r\bto r'
\justifies
\lam{a{:}\tya}r\bto\lam{a{:}\tya}r'
\using\rulefont{cngl}
\end{prooftree}
\\[4ex]
\begin{prooftree}
r\bto r'\quad s\bto s'
\justifies
\letbox{X}{s}{r}\bto \letbox{X}{s'}{r'}
\using\rulefont{cnge}
\end{prooftree}
\\[4ex]
\begin{prooftree}
\phantom{h}
\justifies
\tf{isapp}\,\Box(r'r)\bto \top
\using\rulefont{isapp\top}
\end{prooftree}
&
\begin{prooftree}
(r\text{ not of the form }r'r)
\justifies
\tf{isapp}\,\Box(r)\bto\bot
\using\rulefont{isapp\bot}
\end{prooftree}
\end{array}
$$
\caption{Reduction rules for the modal system}
\label{fig.reduction}
\end{figure}

\begin{rmrk}
We do not have a rule that if $r\bto r'$ then $\Box r\bto \Box r'$.
This would be wrong because it does not respect the integrity of the syntax of a term; syntax, in denotation, does not inherently reduce.

We do however allow reduction under a $\lambda$.
This is purely a design choice; we are interested in making as many terms as possible $\beta$-convertible, and less immediately interested in this paper in finding nice notions of $\beta$-normal form.
If we did not have a denotational semantics then we might have to be more sensitive to such questions (because normal forms are important for consistency)---because we \emph{do} have a denotational semantics, we obtain consistency via soundness and the precise notion of normal form is not so vital. 
\end{rmrk}

\begin{prop}
\label{prop.type.soundness}
If $\Gamma\cent r:\tya$ and $r\to r'$ then $\Gamma\cent r':\tya$.
\end{prop}
\begin{proof}
By a routine induction on $r$.
The case of \rulefont{\beta} $(\lam{a{:}\tya}r)r'\bto r[a{\ssm}r']$ follows by Lemma~\ref{lemm.typesub.a}; that of \rulefont{\beta_\Box} 
follows by Proposition~\ref{prop.typesub}.
\end{proof}

\begin{thrm}
\label{thrm.red.sound}
Suppose $\Gamma\cent r:\tya$ and $\Gamma\cent \varsigma$.
Suppose $r\bto r'$.
Then $\hdenot{\varsigma}{r}=\hdenot{\varsigma}{r'}$. 
\end{thrm}
\begin{proof}
By induction on the derivation of $r\bto r'$.
\begin{itemize*}
\item
\emph{The case of \rulefont{\beta}}\ follows by Lemmas~\ref{lemm.sub.var} and~\ref{lemm.suba.to.varsigma}.
\item
\emph{The case of \rulefont{\beta_\Box}}\ follows by Lemmas~\ref{lemm.letbox.denot} and~\ref{lemm.Xsub.denot}.
\end{itemize*}
\end{proof}

\section{Syntax and typing of the system with contextual types}
\label{sect.contextual.system}

The modal type system is beautiful, but is a little too weak for some applications. 
The issue is that $X$ ranges over \emph{closed} syntax.
If we are working under some $\lambda$-abstractions, we may well find this limiting; we want to work with \emph{open} syntax so that we can refer to the enclosing binder.
This really matters, because it affects the programs we can write.
For instance in the example of exponentiation from Subsection~\ref{subsubsect.exp}, the issue of working under a $\lambda$-abstraction forced us to generate unwanted $\beta$-redexes. 

The contextual system is one way to get around this.
Syntax is still closed, but the notion of closure is liberalised by introducing a context into the modality; to see the critical difference, compare the \rulefont{[]I} rule in Figure~\ref{fig.cmtt.types} with the \rulefont{\Box I} rule from Figure~\ref{fig.modal.types}.
The interested reader can see how this allows us to write a nicer program for exponentiation, which does not generate $\beta$-redexes, in Subsection~\ref{subsubsect.exp.c}.

\subsection{Syntax of the contextual system}

\begin{nttn}
The contextual system needs many vectors of types and atoms-and-types.
For clarity, we write these vectors subscripted, for instance:
\begin{itemize*}
\item
$(a_i{:}\tya_i)_1^n$ is shorthand for $\{a_1{:}\tya_1,\dots,a_n{:}\tya_n\}$.
\item
$[\tya_i]_1^n\tya$ is shorthand for $[\tya_1,\dots,\tya_n]\tya$.
\item
$(\tya_i)_1^n\fto\tya$ is shorthand for $\tya_1\fto(\tya_2\fto\dots(\tya_n\fto\tya))$.
\item
$\{a_i\}_1^n$ is shorthand for $\{a_1,\dots,a_n\}$.
\item
$\lam{(x_i{:}\tya_i)_1^n}r$ is shorthand for $\lam{x_1{:}\tya_1}\dots\lam{x_n{:}\tya_n}r$.
\item
$[a_i{\ssm}x_i]_1^n$ will be shorthand for the map taking $a_i$ to $x_i$ for $1{\leq}i{\leq}n$ and undefined elsewhere (Definition~\ref{defn.substitution.c}).
\end{itemize*}
We may omit the interval where it is understood or irrelevant, so for instance $\{a_i\}$ and $\{a_i\}_i$ are both shorthand for the same thing: ``$\{a_1,\dots,a_n\}$ for some $n$ whose precise value we will never need to reference'', and $(\tya_i)\fto\tya$ is shorthand for ``$(\tya_i)_1^n\fto\tya$ for some $n$ whose precise value we will never need to reference''.
\end{nttn}

We take atoms and unknowns as in Definition~\ref{defn.atoms}.
\begin{defn}
\label{defn.types.c}
Define \deffont{types} inductively by:
\begin{frameqn}
\begin{array}{r@{\ }l}
\tya::=& o \mid \nat \mid \tya\to \tya \mid [\tya_i]_1^n\tya
\end{array}
\end{frameqn}
\end{defn}
$o$ (truth-values),\ $\nat$ (numbers),\ and $\tya\fto\tyb$ (functions) are as in Definition~\ref{defn.types}.
$[\tya_i]_1^n\tya$ is a \emph{contextual type}.
Think of this as generalising the modal types of Definition~\ref{defn.types.c} by `allowing bindings in the box'.
 
\begin{defn}
Fix a set of \deffont{constants} $C$ to each of which is assigned a type $\f{type}(C)$.
We write $C:\tya$ as shorthand for `$C$ is a constant and $\f{type}(C)=\tya$'.
We insist that constants include the following:
$$
\bot:o
\qquad
\top:o
\qquad
\tf{isapp}_\tya:(\Box \tya)\fto o
$$
We may omit the type subscripts where they are clear from context or do not matter.
\end{defn}

\begin{defn}
\label{defn.terms.c}
Define \deffont{terms} inductively by:
\begin{frameqn}
r::= C \mid a \mid \lam{a{:}\tya}r \mid rr \mid [a_i{:}\tya_i]r \mid X\at (r_i)_1^n \mid \letbox{X}{r}{r}
\end{frameqn}
\end{defn}

\begin{rmrk}
\label{rmrk.transfer.to.c}
The syntax of the modal type system in Definition~\ref{defn.terms} injects naturally into that of Definition~\ref{defn.terms.c}, if we map $\Box \text{-}$ to $[\,]\text{-}$ (the empty context) and $\text{-}_\at$ to $\text{-}\at()$.

The important extra complexity is in $X\at(r_i)_1^n$; when $X$ is instantiated by a substitution $\theta$, this triggers an atoms-substitution of the form $[a_i{\ssm}r_i]_1^n$.
See Definition~\ref{defn.sub.c}.
\end{rmrk}

\begin{defn}
\label{defn.hole.free.atoms.c}
Define \deffont{free atoms} $\fa(\rtm)$ and \deffont{free unknowns} $\fv(\rtm)$ by:
\begin{displaymath}
\begin{array}{r@{\ }l@{\qquad}r@{\ }l}
\fa(C) = & \varnothing
&
\fa(a) = & \{ a \}                               
\\
\fa(\lam{a{:}\tya}r) = & \fa(r) \setminus \{ a \}
&
\fa(\rtm\stm) = & \fa(\rtm) \cup \fa(\stm) 
\\
\fa([a_i{:}\tya_i]_{1}^n r) = & \fa(r)\setminus\{a_1,\dots,a_n\}
&
\fa(\letbox{X}{s}{r}) = & \fa(r)\cup\fa(s)
\\
\fa(X\at(s_i)_i) = & \bigcup_i \fa(s_i)
\\[2ex]
\fv(C) = & \varnothing
&
\fv(a) = & \varnothing
\\
\fv(\lam{a{:}\tya}\rtm) = & \fv(\rtm) 
&
\fv(\rtm\stm) = & \fv(\rtm) \cup \fv(\stm)  
\\
\fv([a_i{:}\tya_i]r) = & \fv(r)
&
\fv(\letbox{X}{s}{r}) = & (\fv(r){\setminus}\{X\})\cup\fv(s)
\\
\fv(X\at(s_i)_i) = & \{X\}\cup\bigcup_i\fv(s_i)
\end{array}
\end{displaymath}
\end{defn}

\begin{defn}
We take $a$ to be bound in $r$ in $\lam{a{:}\tya}r$ and $a_1,\dots,a_n$ to be bound in $r$ in $[a_i{:}\tya]_1^n r$,\ and we take $X$ to be bound in $r$ in $\letbox{X}{s}{r}$.
We take syntax up to $\alpha$-equivalence as usual.
For example:
\begin{itemize*}
\item
$\lam{a{:}\tya}a=\lam{b{:}\tya}b$
\item
$\lam{a{:}\tya}[b{:}\tyb]((X\at(b))a)=\lam{b{:}\tya}[a{:}\tyb]((X\at a)b)\neq\lam{b{:}\tya}[b{:}\tyb]((X\at(b))b)$
\item
$\begin{array}[t]{l}
\letbox{X}{[a{:}\tya]a}{(X\at(b))}
\\
=\letbox{Y}{[a{:}\tya]a}{(Y\at (b))}
\\
=\letbox{Y}{[b{:}\tya]b}{(Y\at(b))}
\end{array}$
\end{itemize*}
\end{defn}

\begin{figure}[t]
$$
\begin{array}{c@{\quad}c}
\begin{prooftree}
\phantom{h}
\justifies
\Gamma,a:\tya\cent a:\tya
\using\rulefont{Hyp}
\end{prooftree}
&
\begin{prooftree}
\phantom{h}
\justifies
\Gamma\cent C:\f{type}(C)
\using\rulefont{Const}
\end{prooftree}
\\[2em]
\begin{prooftree}
\Gamma,a{:}\tya\cent r:\tyb 
\justifies
\Gamma\cent (\lam{a{:}\tya}r):\tya\to \tyb
\using\rulefont{{\to}I}
\end{prooftree}
&
\begin{prooftree}
\Gamma\cent r':\tya\to \tyb
\quad
\Gamma\cent r:\tya
\justifies
\Gamma\cent r'r:\tyb
\using\rulefont{{\to}E}
\end{prooftree}
\\[2em]
\begin{prooftree}
\Gamma,(a_i{:}\tya_i)_i\cent r:\tya 
\quad (\fa(r){\subseteq}\{a_i\}_i)
\justifies
\Gamma\cent [a_i{:}\tya_i]r:[\tya_i]\tya
\using\rulefont{[\,]I}
\end{prooftree}
&
\begin{prooftree}
\Gamma,X{:}[\tya_i]\tya\cent r {:} \tyb \ \ \Gamma\cent s{:}[\tya_i]\tya 
\justifies
\Gamma\cent \letbox{X}{s}{r}:\tyb
\using\rulefont{[\,]E}
\end{prooftree}
\\[2em]
\begin{prooftree}
\Gamma,X:[\tya_i]_1^n\tya \cent r_j:\tya_j\quad (1{\leq}j{\leq}n)
\justifies
\Gamma,X:[\tya_i]_1^n\tya \cent X\at(r_i)_1^n:\tya
\using\rulefont{Ext}
\end{prooftree}
\end{array}
$$
\caption{Contextual modal type theory typing rules}
\label{fig.cmtt.types}
\end{figure}

\subsection{Typing for the contextual system}

\begin{defn}
\label{defn.typing.rules.c}
A \deffont{typing} is a pair $a:\tya$ or $X:[\tya_i]_i\tya$.
A \deffont{typing context} $\Gamma$ is a finite partial function from $\atoms\cup\unknowns$ to types (as in Definition~\ref{defn.typing.rules}, except that unknowns have contextual types instead of just box types $\Box\tya$).

A \deffont{typing sequent} is a tuple $\Gamma\cent r:\tya$ of a typing context, a term, and a type.
\begin{frametxt}
Define the \deffont{valid typing sequents} of the contextual modal type system by the rules in Figure~\ref{fig.cmtt.types}.
\end{frametxt}
\end{defn}

Recall the notation $\Gamma|_U$ from Notation~\ref{nttn.restrict}.
Proposition~\ref{prop.ws.c} repeats Proposition~\ref{prop.ws} for the contextual system:
\begin{prop}
\label{prop.ws.c}
If $\Gamma\cent r:\tya$ and $\Gamma'|_{\fv(r)\cup\fa(r)}=\Gamma|_{\fv(r)\cup\fa(r)}$ then $\Gamma'\cent r:\tya$.
\end{prop}
\begin{proof}
By a routine induction on $r$.
\end{proof}

\subsection{Substitution}

Definition~\ref{defn.substitution.c} reflects Definition~\ref{defn.substitution} for the richer syntax of terms:
\begin{defn}
\label{defn.substitution.c}
An \deffont{(atoms-)substitution} $\sigma$ is a finite partial function from atoms $\mathbb A$ to terms.
$\sigma$ will range over atoms-substitutions.

Write $\dom(\sigma)$ for the set ${\{a\mid\sigma(a)\text{ defined}\}}$ 

Write $\id$ for the \deffont{identity} substitution, such that $\dom(\sigma)=\varnothing$.
 
Write $[a_i{\ssm}x_i]_1^n$ for the map taking $a_i$ to $x_i$ for $1{\leq}i{\leq}n$ and undefined elsewhere.

An \deffont{(unknowns-)substitution} $\theta$ is a finite partial function from unknowns $\unknowns$ to terms such that if $\theta(X)$ is defined then $\theta(X)=[a_i{:}\tya_i]_1^n r$ for some $r$ with $\fa(r)\subseteq\{a_1,\dots,a_n\}$ (so $\fa(\theta(X))=\varnothing$ for every $X\in\dom(\theta)$). 

$\theta$ will range over unknowns-substitutions.

We write $\dom(\theta)$,\ $\id$,\ and $[X_i{\ssm}t_i]_1^n$ just as for atoms-substitutions (we will be most interested in the case that $n=1$).
\end{defn}

We also reflect Definition~\ref{defn.fa.sigma} and write $\fa(\sigma)$ and $\fv(\theta)$, but using the notions of `free atoms' and `free unknowns' from Definition~\ref{defn.hole.free.atoms.c}.
The definition is formally identical:
$$
\begin{array}{r@{\ }l}
\fa(\sigma)=&\dom(\sigma)\cup\{\fa(\sigma(a))\mid a\in\dom(\sigma)\}
\quad\text{and}
\\
\fv(\theta)=&\dom(\theta)\cup\{\fv(\theta(X))\mid X\in\dom(\theta)\}
\end{array}
$$

\begin{frametxt}
\begin{defn}
\label{defn.sub.c}
Define substitution actions $r\sigma$ and $r\theta$ by the rules in Figure~\ref{fig.sub.c}.
\end{defn}
\end{frametxt}

\begin{rmrk}
The capture-avoidance side-conditions of Definition~\ref{defn.sub.c} (of the form `$\ast\not\in\fa(\sigma)$' or `$\ast\not\in\fv(\theta)$') can be guaranteed by $\alpha$-renaming.

Strictly speaking the case of $(X\at(r_i)_1^n)\theta$ introduces a partiality into the notion of substitution action; we assume that $\theta(X)=[a_i{:}\tya_i]_1^m s'$ and \emph{for this to make sense} it must be that $n=m$; if $n\neq m$ then the definition is not well-defined.
However, for well-typed syntax this is guaranteed not to happen, and since this is the only case we will care about, we will never notice this.
\end{rmrk} 

\begin{figure}[t]
$$
\hspace{-.5ex}
\begin{array}{r@{\ }l@{\hspace{-1ex}}r@{\ }l@{\ \ }l}
C\sigma=&C
&
a\sigma=&\sigma(a) & (a\in\dom(\sigma))
\\
(rs)\sigma=&(r\sigma)(s\sigma)
&
a\sigma=&a & (a\not\in\dom(\sigma))
\\
(X\at(r_i)_i)\sigma=&X\at(r_i\sigma)_i
&
(\lam{c{:}\tya}r)\sigma=&\lam{c{:}\tya}(r\sigma) &
(c\not\in\fa(\sigma))
\\
(\letbox{Y}{s}{r})\sigma=&\letbox{Y}{s\sigma}{r\sigma}
&
([a_i{:}\tya_i]r)\sigma=&[a_i{:}\tya_i](r\sigma) 
& (a_i\not\in\fa(\sigma)\ \text{all}\ i) 
\\[2.5ex]
C\theta=&C
&
a\theta=&a
\\
(rs)\theta=&(r\theta)(s\theta)
&
(X\at(r_i)_i)\theta=&s'[a_i{\ssm}r_i] 
&(\theta(X){=}[a_i{:}\tya_i]s')
\\
([a_i{:}\tya_i]r)\theta=&[a_i{:}\tya_i](r\theta)
&
(X\at(r_i))\theta=&X\at(r_i)
&(X\not\in\dom(\theta))
\\
(\lam{c{:}\tya}r)\theta=&\lam{c{:}\tya}(r\theta) 
&
(\letbox{Y}{s}{r})\theta=&\letbox{Y}{s\theta}{r\theta}
&(Y\not\in\fv(\theta))
\end{array}
$$
\caption{Substitution actions for atoms and unknowns (contextual syntax)}
\label{fig.sub.c}
\end{figure}

We conclude this section with some important definitions and results about the interaction of substitution and typing, which will be needed for Theorem~\ref{thrm.type.soundness.c}.

Definition~\ref{defn.gamma.theta.c} reflects Definition~\ref{defn.gamma.theta}, but we need $\Gamma\cent\sigma$ as well as $\Gamma\cent\theta$:
\begin{defn}
\label{defn.gamma.theta.c}
Write $\Gamma\cent\theta$ when 
if $X\in\dom(\theta)$ then $X{:}[\tya_i]\tya\in\Gamma$ for some $[\tya_i]\tya$ and $\Gamma\cent\theta(X):[\tya_i]\tya$.

Similarly write $\Gamma\cent\sigma$ when 
if $a\in\dom(\sigma)$ then $a{:}\tya\in\Gamma$ for some $\tya$ and $\Gamma\cent \sigma(a):\tya$.
\end{defn}

\begin{lemm}
\label{lemm.fa.rtheta.c}
$\fa(r\theta)=\fa(r)$ where $r\theta$ is defined.
\end{lemm}
\begin{proof}
By a routine induction on $r$ using our assumption of Definition~\ref{defn.substitution.c} that if $X\in\dom(\theta)$ then $\fa(\theta(X))=\varnothing$.
\end{proof}
 
Lemma~\ref{lemm.sigma.c} reflects Lemma~\ref{lemm.typesub.a}.
However, unlike was the case for the modal system, it is needed for Proposition~\ref{prop.rtheta.c}/\ref{prop.rtheta} because the case of $(X\at(r_i))\theta$ in Definition~\ref{defn.sub.c} triggers an atoms-substitution.
\begin{lemm}
\label{lemm.sigma.c}
Suppose $\Gamma\cent r:\tya$ and $\Gamma\cent\sigma$. 
Then $\Gamma\cent r\sigma:\tya$. 
\end{lemm}
\begin{proof}
By routine inductions on the derivation of $\Gamma\cent r:\tya$.
\end{proof}

Proposition~\ref{prop.rtheta.c} reflects Proposition~\ref{prop.rtheta} and is needed for soundness of the denotation.
The proof is significantly more complex, because of the atoms-substitution that can be introduced by the case of $(X\at(s_j))\theta$.
This is handled in the proof below using Lemma~\ref{lemm.sigma.c}.
\begin{prop}
\label{prop.rtheta.c}
Suppose $\Gamma\cent r:\tya$ and $\Gamma\cent \theta$.
Then $\Gamma\cent r\theta:\tya$.
\end{prop}
\begin{proof}
By a routine induction on the typing of $r$.
We consider two cases:
\begin{itemize*}
\item
\emph{The case of \rulefont{[\,] I}.}\quad
Suppose $\Gamma,(b_j{:}\tyb_j)\cent r:\tya$ and $\fa(r){\subseteq}\{b_j\mid j\}$ so that $\Gamma\cent [b_j{:}\tyb_j]r:[\tyb_j]\tya$ by \rulefont{[\,] I}.
By inductive hypothesis $\Gamma,(b_j{:}\tyb_j)\cent r\theta:\tya$.
By Lemma~\ref{lemm.fa.rtheta.c} $\fv(r\theta){\subseteq}\{b_j\mid j\}$. 
We use \rulefont{[\,]I} and the fact that $([b_j{:}\tyb_j]r)\theta=[b_j{:}\tyb_j](r\theta)$.
\item
\emph{The case of \rulefont{Ext} for $X\in\dom(\theta)$.}\quad
Suppose $\Gamma,X{:}[\tya_j]_1^m\tya\cent s_j:\tya_j$ for each $1{\leq}j{\leq}m$ so that by \rulefont{Ext} $\Gamma,X{:}[\tya_j]_j\tya\cent X\at(s_j)_j{:}\tya$.
By inductive hypothesis $\Gamma\cent s_j\theta:\tya_j$ for each $j$.
By assumption $\varnothing\cent \varsigma(X):[\tya_j]_j\tya$, which implies that $\varsigma(X)=[a_j{:}\tya_j]r'$ for some $r'$ such that $(a_j{:}\tya_j)_j\cent r':\tya$.
By Lemma~\ref{lemm.sigma.c} $\Gamma\cent r'[a_j{\ssm}s_j\theta]:\tya$.
By the definitions $(X\at(s_j)_j)\theta=r'[a_j{\ssm}s_j]_j$, so we are done.
\end{itemize*}
\end{proof}

We could now give a theory of reduction for the contextual system, following the definition of reduction for the modal system in Subsection~\ref{subsect.reduction}.
However, we will skip over this; the interested reader is referred elsewhere \cite{nan+pfe:jfp05}.
What is more interesting, from the point of view of this paper, is the models we define for the contextual system, which we come to next.

\section{Contextual models}
\label{sect.contextual.models}

\subsection{Denotational semantics}

\maketab{tab3}{@{\hspace{-0em}}R{12em}@{\ }L{23em}}

Definition~\ref{defn.interpret.types.c} is like Definition~\ref{defn.interpret.types}, except that instead of box types, we have contextual types:
\begin{defn}
\label{defn.interpret.types.c}
Define $\hdenot{}{\tya}$ the \deffont{interpretation} of types by induction in Figure~\ref{fig.denot.types.c}. 
\end{defn}

\begin{figure}[t]
\begin{minipage}{\textwidth}
\begin{tab3}
\hdenot{}{o}=&\{\top^\hden,\bot^\hden\}
\\
\hdenot{}{\nat}=&\{0,1,2,\dots\}
\\
\hdenot{}{\tya\fto\tyb}=&\hdenot{}{\tyb}^{\hdenot{}{\tya}}
\\
\hdenot{}{[\tya_i]_1^n\tya}=&
\{[a_i{:}\tya_i]_1^n r\mid \varnothing\cent [a_i{:}\tya_i]_1^n r:[\tya_i]_1^n \tya\}
\times
\hdenot{}{\tya}^{\Pi_{i{=}1}^n\hdenot{}{\tya_i}}
\end{tab3}
\end{minipage}
\caption{Denotational semantics for CMTT types}
\label{fig.denot.types.c}
\end{figure}
\begin{figure}[t]
\begin{minipage}{\textwidth}
\begin{tab3}
\hdenot{\varsigma}{\top}=&\top^\hden
\\
\hdenot{\varsigma}{\bot}=&\bot^\hden
\\
\hdenot{\varsigma}{a}=&\varsigma(a)
\\
\hdenot{\varsigma}{\lam{a{:}\tya}r}=&(x{\in}\hdenot{}{\tya}\mapsto \hdenot{\varsigma[a{\ssm}x]}{r})
\\
\hdenot{\varsigma}{r'r}=&\hdenot{\varsigma}{r'}\,\hdenot{\varsigma}{r}
\\
\hdenot{\varsigma}{[a_i{:}\tya_i]_1^n r}=&[a_i{:}\tya_i]_1^n(r\,\varsigma_\unknowns) :: \bigl(\lam{(x_i{\in}\hdenot{}{\tya_i})_1^n}\hdenot{\varsigma[a_i{\ssm}x_i]_1^n}{r}\bigr) 
\\
\hdenot{\varsigma}{X\at(r_i)_1^n}=&\tl(\varsigma(X))\,(\hdenot{\varsigma}{r_i})_1^n 
\\
\hdenot{\varsigma}{\letbox{X}{s}{r}}=&\hdenot{\varsigma[X{\ssm}\hdenot{\varsigma}{s}]}{r}
\\
\hdenot{\varsigma}{\tf{isapp}_\tya}([a_i{:}\tya_i](r'r))=& \top^\hden 
\\
\hdenot{\varsigma}{\tf{isapp}_\tya}([a_i{:}\tya_i](r))=&\bot^\hden \quad\text{otherwise} 
\end{tab3}
\end{minipage}
\caption{Denotational semantics for terms of the contextual system}
\label{fig.denot.terms.c}
\end{figure}

\begin{defn}
\label{defn.valuation.c}
A \deffont{valuation} $\varsigma$ is a finite partial function on $\atoms\cup\unknowns$.

We define $\varsigma[X{\ssm}x]$ and $\varsigma[a{\ssm}x]$ just as in Definition~\ref{defn.valuation}.
\end{defn}

\begin{defn}
\label{defn.varsigmaunknowns.c}
Write $\varsigma_\unknowns$ for the substitution (Definition~\ref{defn.substitution.c}) such that $\varsigma_\unknowns(X)=\hd(\varsigma(X))$ if $\varsigma(X)$ is defined,
and $\varsigma_\unknowns(X)$ is undefined if $\varsigma(X)$ is undefined.
\end{defn}

\begin{defn}
\label{defn.gamma.varsigma.c}
If $\Gamma$ is a typing context then write $\Gamma\cent\varsigma$ when:
\begin{frametxt}
\begin{enumerate*}
\item
$\dom(\Gamma)=\dom(\varsigma)$.
\item
If $a\in\dom(\varsigma)$ then $\Gamma(a)=\tya$ for some $\tya$ and $\varsigma(a)\in\denot{}{}{\tya}$.
\item
If $X\in\dom(\varsigma)$ then $\Gamma(X)=[\tya_i]\tya$ and $\varsigma(X)\in\hdenot{}{[\tya_i]\tya}$.
\end{enumerate*}
\end{frametxt}
\end{defn}

\begin{rmrk}
Unpacking Definition~\ref{defn.interpret.types.c}, clause~3 (the one for $X$) means that $\varsigma(X)=[a_i{:}\tya_i] r'$ and $\varnothing\cent [a_i{:}\tya_i]r':[\tya_i]\tya$.
Following the typing rules of Figure~\ref{fig.cmtt.types}, this is equivalent to $(a_i{:}\tya_i)_i\cent r':\tya$.
\end{rmrk}

\begin{defn}
\label{defn.denot.terms.c}
For each constant $C:\tya$ other than $\top$, $\bot$, and $\tf{isapp}$ fix some interpretation $C^\hden$ which is an element $C^\hden\in\hdenot{}{\tya}$.
Suppose $\Gamma\cent\varsigma$ and $\Gamma\cent r:\tya$.
\begin{frametxt}
An \deffont{interpretation} of terms $\hdenot{\varsigma}{r}$ is defined in Figure~\ref{fig.denot.terms.c}.
\end{frametxt}
\end{defn}

\begin{rmrk}
Definition~\ref{defn.denot.terms.c} is in the same spirit as Definition~\ref{defn.denot.terms}, but now the modal types are contextual; the modal box contains a context $a_1{:}\tya_1,\dots,a_n{:}\tya_n$.
When we calculate $\hdenot{\varsigma}{X\at(r_i)_1^n}$ the denotation of $X\at(r_i)_1^n$, the denotations of the terms $r_i$ provide denotations for the variables in that context.
\end{rmrk}

\begin{lemm}
\label{lemm.varsigma.to.varsigmaun.c}
If $\Gamma\cent\varsigma$ then $\Gamma\cent\varsigma_\unknowns$.
\end{lemm}
\begin{proof}
If $X\not\in\dom(\varsigma)$ then $X\not\in\dom(\varsigma_\unknowns)$.

Suppose $X\in\dom(\varsigma)$.
By Definition~\ref{defn.varsigmaunknowns.c} $\varsigma_\unknowns(X)=\hd(\varsigma(X))$. 
By Definition~\ref{defn.gamma.varsigma.c} $\varsigma_\unknowns(X)\in\hdenot{}{[\tya_i]_1^n\tya}$ for some $[\tya_i]_1^n\tya$.
Unpacking Figure~\ref{fig.denot.types.c} this implies that $\varsigma_\unknowns(X)=[a_i{:}\tya_i]_1^n r$ for some $\varnothing\cent [a_i{:}\tya_i]_1^n r:[\tya_i]_1^n\tya$, and we are done.
\end{proof}

\begin{lemm}
\label{lemm.typesub.c}
Suppose $\Gamma\cent r:\tya$ and $\Gamma\cent \varsigma$.
Then $\Gamma|_\atoms\cent r\varsigma_\unknowns:\tya$.
\end{lemm}
\begin{proof}
By Lemma~\ref{lemm.varsigma.to.varsigmaun.c} $\Gamma\cent\varsigma_\unknowns$.
By Proposition~\ref{prop.rtheta.c} $\Gamma\cent r\,\varsigma_\unknowns:\tya$.
By Lemma~\ref{lemm.fa.rtheta.c} $\fa(r\varsigma_\unknowns)=\fa(r)$.
It is a fact that $\fa(r)\subseteq\dom(\Gamma|_\atoms)$, so by Proposition~\ref{prop.ws} $\Gamma|_\atoms\cent r\,\varsigma_\unknowns:\tya$ as required.
\end{proof}

\begin{frametxt}
\begin{thrm}[Soundness]
\label{thrm.type.soundness.c}
If $\Gamma\cent r:\tya$ and $\Gamma\cent\varsigma$ then $\hdenot{\varsigma}{r}\in\hdenot{}{\tya}$.
\end{thrm}
\end{frametxt}
\begin{proof}
By induction on the the derivation of $\Gamma\cent r:\tya$. 
Most of the rules follow by properties of sets and functions.
We consider the interesting cases:
\begin{itemize*}
\item
\emph{Rule \rulefont{[\,]I}.}\quad
Suppose $\Gamma,(a_i{:}\tya_i)_1^n\cent r:\tya$ so that by \rulefont{[\,]I} $\Gamma\cent [a_i{:}\tya_i]r:[\tya_i]\tya$.
Suppose $\fa(r){\subseteq}\{a_1,\dots,a_n\}$ and $\Gamma\cent\varsigma$.
Using Lemma~\ref{lemm.typesub.c} $\varnothing\cent [a_i{:}\tya_i](r\varsigma_\unknowns):\tya$.

Suppose $x_i\in\hdenot{}{\tya_i}$ for $1{\leq}i{\leq}n$.
By  Definition~\ref{defn.gamma.varsigma.c}
$$
\Gamma,(a_i{:}\tya_i)_1^n \cent \varsigma[a_i{\ssm}x_i]_1^n
$$
so by inductive hypothesis for the derivation of $\Gamma,(a_i{:}\tya_i)_1^n\cent r:\tya$ it follows that 
$$
\hdenot{\varsigma[a_i{\ssm}x_i]_1^n}{r}\in\hdenot{}{\tya} . 
$$
Now this was true for arbitrary $x_i$ and it follows from Definition~\ref{defn.interpret.types.c} that $\hdenot{\varsigma}{[a_i{:}\tya_i]r}\in\hdenot{}{[\tya_i]\tya}$ as required. 
\item
\emph{Rule \rulefont{[\,]E}.}\quad
Suppose $\Gamma,X{:}[\tya_i]\tya\cent r:\tyb$ and $\Gamma\cent s:[\tya_i]\tya$ so that by \rulefont{[\,]E} $\Gamma\cent \letbox{X}{s}{r}:\tyb$.

Suppose $\Gamma\cent\varsigma$.
By inductive hypothesis for $\Gamma\cent s:[\tya_i]\tya$ we have $\hdenot{\varsigma}{s}\in\hdenot{}{[\tya_i]\tya}$.

It follows by Definition~\ref{defn.gamma.varsigma.c} that 
$\Gamma,X{:}[\tya_i]\tya\cent \varsigma[X{\ssm}\hdenot{\varsigma}{s}]$ so by inductive hypothesis for $\Gamma,X{:}[\tya_i]\tya\cent r:\tyb$ we have $\hdenot{\varsigma[X{\ssm}\hdenot{\varsigma}{s}]}{r}\in\hdenot{}{\tyb}$.
We now observe by Definition~\ref{defn.denot.terms.c} that 
$$
\hdenot{\varsigma}{\letbox{X}{s}{r}}=\hdenot{\varsigma[X{\ssm}\hdenot{\varsigma}{s}]}{r}\in\hdenot{}{\tyb}.
$$
\item
\emph{Rule \rulefont{Ext}.}\quad
Suppose 
$\Gamma,X{:}[\tya_i]_1^n\tya\cent r_i:\tya_i$ for $1{\leq}i{\leq}n$ so that by \rulefont{Ext} $\Gamma,X{:}[\tya_i]_1^n\tya\cent X\at (r_i)_1^n{:}\tya$. 

By inductive hypothesis for the typings $\Gamma,X{:}[\tya_i]_1^n\tya\cent r_i:\tya_i$ we have $\hdenot{\varsigma}{r_i}\in\hdenot{}{\tya_i}$ for $1{\leq}i{\leq}n$.

Suppose $\Gamma,X{:}[\tya_i]\tya\cent\varsigma$. 
By Definitions~\ref{defn.gamma.varsigma.c} and~\ref{defn.denot.terms.c} this means that $\varsigma(X)=([a_i{:}\tya_i]_1^n r')::f$ for some $\varnothing\cent [a_i{:}\tya_i] r':[\tya_i]\tya$ and some $f\in(\Pi_{i{=}1}^n\hdenot{}{\tya_i})\fto\hdenot{}{\tya}$.
It follows that $f\,(\hdenot{\varsigma}{r_i})_1^n\in\hdenot{}{\tya}$ as required.
\item
\emph{Rule \rulefont{Hyp}.}\quad
Suppose $\Gamma,a{:}\tya\cent\varsigma$.
By Definition~\ref{defn.gamma.varsigma.c} this means that $\varsigma(a)\in\hdenot{}{\tya}$.
By Definition~\ref{defn.denot.terms.c} $\hdenot{\varsigma}{a}=\varsigma(a)$.
The result follows.
\item
\emph{Rule \rulefont{{\fto}I}.}\quad
Suppose $\Gamma,a{:}\tya\cent r:\tyb$ so that by \rulefont{{\fto}I} $\Gamma\cent\lam{a{:}\tya}r:\tya\fto\tyb$.
Suppose $\Gamma\cent\varsigma$ and choose any $x\in\hdenot{}{\tya}$.
It follows that 
$
\Gamma,a{:}\tya\cent\varsigma[a{\ssm}x]
$
and so by inductive hypothesis that
$
\hdenot{\varsigma[a{\ssm}x]}{r} \in\hdenot{}{\tyb} .
$

Since $x\in\hdenot{}{\tya}$ was arbitrary, by Definition~\ref{defn.denot.terms.c} we have that
$$
\hdenot{\varsigma}{\lam{a{:}\tya}r}
=
(x\in\hdenot{}{\tya}\mapsto \hdenot{\varsigma[a{\ssm}x]}{r}) \in\hdenot{}{\tya\fto\tyb} .
$$
\end{itemize*}
\end{proof}

\begin{corr}
\label{corr.later.c}
\begin{enumerate*}
\item
There is no term $s$ such that $\varnothing\cent s:(\nat\fto\nat)\fto[\,](\nat\fto\nat)$ is typable and such that the map 
$\lam{x{\in}\nat^\nat}\hd(\hdenot{\varnothing}{s}\,x)\in\hd(\hdenot{}{[\,](\nat\fto\nat)})^{\hdenot{}{\nat\fto\nat}}$ is injective.
\item
There is no term $s$ such that $\varnothing\cent s:(\nat\fto\nat)\fto[\nat]\nat$ is typable and such that the map 
$\lam{x{\in}\nat^\nat}\hd(\hdenot{\varnothing}{s}\,x)\in\hd(\hdenot{}{[\nat]\nat})^{\hdenot{}{\nat\fto\nat}}$ is injective.
\end{enumerate*}
\end{corr}
\begin{proof}
$\hd\hdenot{}{[](\nat\fto\nat)}$ and $\hd\hdenot{}{[\nat]\nat}$ are both countable sets whereas $\hdenot{}{\nat\fto\nat}=\nat^\nat$ is uncountable. 
\end{proof}

\subsection{Typings and denotations in the contextual system}

The examples from Subsection~\ref{subsect.some.programs} transfer to the contextual system if we translate $\Box \text{-}$ to $[\,]\text{-}$ and $\text{-}_\at$ to $\text{-}\at()$ (cf. Remark~\ref{rmrk.transfer.to.c}).
So the reader can look to Subsection~\ref{subsect.some.programs} for some simpler examples.

We now consider some slightly more advanced ideas.

\subsubsection{Moving between $[\tya]\tyb$ and $[\,](\tya\fto\tyb)$}

We can move between the types $[\tya]\tyb$ and $[\,](\tya\fto\tyb)$ using terms
$f:[\tya]\tyb\fto[\,](\tya\fto\tyb)$ 
and
$g:[\,](\tya\fto\tyb)\fto[\tya]\tyb$
defined as follows:
$$
\begin{array}{@{\varnothing\cent\ }r@{\ }l@{\ :\ }l}
f =& 
\lam{c{:}[\tya]\tyb}\letbox{X}{c}{[\,]\lam{a{:}\tya} X\at(a)}
&[\tya]\tyb\fto[\,](\tya\fto\tyb)
\\
g =& 
\lam{c{:}[\,](\tya\fto\tyb)}\letbox{X}{c}{[a{:}A]((X\at())a)}
&[\,](\tya\fto\tyb)\fto[\tya]\tyb
\end{array}
$$
It is routine to check that the typings above are derivable using the rules in Figure~\ref{fig.cmtt.types}.

Intuitively, we can write the following:
\begin{itemize*}
\item
$f$ maps $[a{:}\tya]r$ to $[\,]\lam{a{:}\tya}r$. 
\item
$g$ maps $[\,]\lam{a{:}\tya}r$ to $[a{:}\tya]((\lam{a{:}\tya}r)a)$\ \ (so $g$ introduces an $\beta$-redex). 
\end{itemize*}
This can be made formal as follows:
$$
\begin{array}{r@{\ }l}
\hd\hdenot{\varsigma}{f\,([a{:}\tya]r)}=&[\,]\lam{a{:}\tya}(r\varsigma_\unknowns)
\qquad\quad\text{and}
\\
\hd\hdenot{\varsigma}{g\,([\,]\lam{a{:}\tya}r)}=&[a{:}\tya]((\lam{a{:}\tya}(r\varsigma_\unknowns))a)
\end{array}
$$
The fact that $g$ introduces a $\beta$-redex reflects the fact that we have given our language facilities to build up syntax---but not to destroy it. 
We can build a precise inverse to $f$ if we give ourselves an explicit destructor for $\lambda$-abstraction.

So for instance, we can give ourselves option types and then admit a constant symbol $\tf{match\_lam}:[\,](\tya\fto \tya) \fto \tf{option} ([\tya]\tyb)$ with intended behaviour as follows: 
\[
\tf{match\_lam}\,(t) = \left \{
\begin{array}{ll}
  \tf{some}\ ([a{:}A]\,r) & \mbox{if $\hdenot{}{t} = ([\ ]\lambda a{:}A. r) {::} \_$}\\
  \tf{none} & \mbox{otherwise}
\end{array}\right.
\]
Using $\tf{match\_lam}$ we could map from $[\,](\tya\fto\tyb)$ to $[\tya]\tyb$ in a manner that is inverse to $f$.\footnote{We do not promote this language directly as a practical programming language, any more than one would promote the pure $\lambda$-calculus. 
We should add constants for the operations we care about.

The point is that in this language, there are things we can do using the modal types that cannot be expressed directly in the pure $\lambda$-calculus, no matter how many constants we might add. 
}

\subsubsection{The example of exponentiation, revisited}
\label{subsubsect.exp.c}

Recall from Subsection~\ref{subsubsect.exp} the discussion of exponentiation and how in the modal system the natural term to meta-program exponentiation introduced $\beta$-reducts.

The following term implements exponentiation:
$$
\begin{array}[t]{r@{\ }l}
\f{exp}\,0\Rightarrow\ & [b{:}\nat]1 
\\ 
\f{exp}\,(\tf{succ}\,n)\Rightarrow\ &\letbox{X}{[b{:}\nat]\f{exp}\,n}{[b{:}\nat](b*(X\at(b)))} 
\end{array}
$$
This term does not generate $\beta$-reducts in the way we noted of the corresponding term from Subsection~\ref{subsubsect.exp}. 
For instance, 
$$
\hd\hdenot{\varnothing}{\f{exp}\,2} = [b{:}\nat](b*b*1) .
$$
Compare this with Subsection~\ref{subsubsect.exp.denot}.

Think of the $[b{:}\nat]$ in $[b{:}\nat]r$ as a `translucent lambda', and think of $X\at(r_i)$ as a corresponding application.
We can use these to carry out computation---a rather weak computation; just a few substitutions as formalised in the clause for $X\at(r_i)_i$ in Figure~\ref{fig.sub.c}---but this computation occurs \emph{inside} a modality, which we could not do with an ordinary $\lambda$-abstraction.

Now might be a good moment to return to the clause for $[a_i{:}\tya_i]r$ in Figure~\ref{fig.denot.terms.c}:
$$
\hdenot{\varsigma}{[a_i{:}\tya_i]_1^n r}=[a_i{:}\tya_i]_1^n(r\,\varsigma_\unknowns) :: \bigl(\lam{(x_i{\in}\hdenot{}{\tya_i})_1^n}\hdenot{\varsigma[a_i{\ssm}x_i]_1^n}{r}\bigr) 
$$
We see the $\lambda$-abstraction in the semantics, and we also see its `translucency': the $\lambda$-abstraction appears in the extension, but is also associated with a non-functional intension.


\subsubsection{Syntax to denotation}
\label{subsect.code.to.values}

\maketab{tab2}{@{\hspace{-2em}}R{8em}@{\ }L{20em}@{}L{12em}}
There is a schema of \emph{unpack} programs, parameterised over $(a_i{:}\tya_i)_1^n$ which evaluates syntax with $n$ free atoms:
\begin{tab2}
\f{unpack}=&\lam{b{:}[\tya_i]_1^n\tyb}\letbox{X}{b}\lam{(a_i{:}\tya_i)_1^n}X\at (a_i)_1^n
&:[\tya_i]_1^n\tyb\fto((\tya_i)_1^n\fto\tyb)
\end{tab2}
We can express the following connection between $\f{unpack}$ (which is a term) and $\tl$ (which is a function on denotations):
\begin{lemm}
\label{lemm.unpack.routine}
Suppose $\Gamma\cent [a_i{:}\tya_i]s:[\tya_i]\tya$ and $\Gamma\cent\varsigma$.
Then 
$$
\hdenot{\varsigma}{\f{unpack}\,[a_i{:}\tya_i]s}=\tl\hdenot{\varsigma}{[a_i{:}\tya_i]s} .
$$
\end{lemm}
\begin{proof}
By long but routine calculations unpacking Figure~\ref{fig.denot.terms.c}.
\end{proof}

As an aside, note that if we have diverging terms $\omega_i:\tya_i$ then we can combine this with $\f{unpack}$ to obtain a term $\varnothing\cent \lam{a{:}[\tya_i]\tya}\f{unpack}\,a\,(\omega_i):[\tya_i]\tya \fto \tya$.
In a call-by-name evaluation strategy, this loops forever if evaluation tries to refer to one of the (diverging) arguments.

\subsubsection{Modal-style axioms}
\label{subsect.modal-style.axioms}

As in Subsection~\ref{subsubsection.one.line} we can write functions corresponding to axioms from the necessity fragment of S4:
\begin{tab2}
T =& \lam{a{:}[\,]\tya}\letbox{X}{a}{X \at ()}
&: [\,]\tya\fto\tya
\\
4 =& \lambda x.\letbox{X}{x}{[\ ][\ ] X \at ()}
&:  [\ ]A \fto [\ ][\ ]A  
\\
K =& \lam{f}\lam{x}\letbox{F}{f}{\letbox{X}{x}{F \at ()\ X \at ()}}
& : [\ ](A \fto B) \fto [\ ]A \fto [\ ]B
\end{tab2}
(Of course, $T$ is just a special case of $\f{unpack}$ above.)

\subsubsection{More general contexts}

Versions of the terms $4$ and $K$ exist for non-empty contexts. 
For example, we can have a schema of $4_\Gamma$ axioms, for any context $\Gamma$:
\begin{tab2}
4_{\Gamma} = & \lam{x{:}[\Gamma]\tya}\letbox{X}{x}{[\ ][\Gamma]X \at(\id_\Gamma)}
&:  [\Gamma] A \fto [\ ][\Gamma] A 
\end{tab2}
Here and below we abuse notation by putting $[\Gamma]$ in the type; we intend the \emph{types} in $\Gamma$, with the variables removed.

Above, $\id_\Gamma$ is the identity substitution defined inductively on
$\Gamma$ by
$$
\id_{\cdot}  =  \cdot
\qquad\text{and}\qquad
\id_{\Gamma, x{:}A} =  \id_\Gamma, x .
$$
Note that the terms realising $4_\Gamma$ are not uniform, because the substitution $\id_\Gamma$ is not a term in the language; it is a
meta-level concept, producing different syntax depending on $\Gamma$.

Similarly, we have a schema of $K_\Gamma$ terms:
\begin{tab2}
K_\Gamma =&  \lam{f} \lam{x} \letbox{F}{f}{\letbox{X}{x}{[\Gamma] F \at \id_{\Gamma}\,X \at \id_{\Gamma}}}
&:  [\Gamma] (A \fto B) \fto [\Gamma] A \fto [\Gamma] B  
\end{tab2}
\dots and terms exposing the structural rules of contexts:
\begin{tab2}
\mathit{weaken}_{\Gamma_1, \Gamma_2} =&  \lambda z. \letbox{Z}{z}{[\Gamma_1, \Gamma_2] (Z \at (\id_{\Gamma_1}))}
&: [\Gamma_1] A \fto [\Gamma_1, \Gamma_2] A 
\\
\mathit{contract}_{B} = & \lambda z. \letbox{Z}{z}{[x{:}B] (Z \at (x,x))}
&:  [B, B] A \fto [B] A 
\\
\mathit{exchange}_{B,C}  = & \lambda z. \letbox{Z}{z}{[y{:}C, x{:}B] (Z \at (x, y))}
&:  [B, C] A \fto [C, B] A 
\end{tab2}
We give $\mathit{weaken}$ in full generality and then for brevity $\mathit{contract}$ and $\mathit{exchange}$ only for two-element contexts.
If we think in terms of multimodal logic \cite{gabbay:mandml} these terms `factor', `fuse', and `rearrange' contexts/modalities.

\section{Shapeliness}
\label{sect.shapely}

We have seen semantics to both the modal and contextual type systems.
We have also noted that, like function-spaces, our semantics \emph{inflates}.
We discussed why in Remark~\ref{rmrk.why.inflate} and Subsection~\ref{subsubsect.explain}.

In this section we delve deeper into the fine structure of the denotation to isolate a property of those parts of the denotation that can be described by syntax (Definition~\ref{defn.shape}).
This is an attractive well-formedness/well-behavedness property in its own right, and furthermore, we can exploit it to strengthen Corollaries~\ref{corr.later} and~\ref{corr.later.c} (see Corollary~\ref{corr.more.c}).

\begin{defn}
\label{defn.shape}
Define the \deffont{shapely} $x\in\hdenot{}{\tya}$ inductively by the rules in Figure~\ref{fig.shape}.

Call $\varsigma$ \deffont{shapely} when:
\begin{itemize*}
\item
$\varsigma(X)$ is shapely for every $X\in\dom(\varsigma)$.
\item
$\varsigma(a)$ is shapely for every $a\in\dom(\varsigma)$.
\end{itemize*}
\end{defn}

Intuitively, $x$ is shapely when, if it is intensional (so $x$ is in some $\hdenot{}{[\tya_i]\tya}$) then the intension $\hd(x)$ and the extension $\tl(x)$ match up.
In particular, this means that elements in $\hdenot{}{\mathbb B}$, $\hdenot{}{\nat}$, or $\hdenot{}{\nat\fto\nat}$---are automatically shapely. 
Conversely, $x$ is not shapely if it has an intension and an extension and they do not match up.
The paradigmatic non-shapely element is $[\,]0::1$, since the intension `the syntax $0$' does not match the extension `the number $1$'.

\begin{figure}[t]
$$
\begin{gathered}
\begin{prooftree}
\tl(x)\in\hdenot{}{\tya} \ \text{is shapely}
\quad
x=\hdenot{\varnothing}{\hd(x)}
\justifies
x\in\hdenot{}{[\,]\tya}\ \text{is shapely}
\using\rulefont{Shape[\,]}
\end{prooftree}
\\[2ex]
\begin{prooftree}
\Forall{y\in\hdenot{}{\tyb}}y\ \text{is shapely}\limp xy\in\hdenot{}{\tya}\ \text{is shapely}
\justifies
x\in\hdenot{}{\tyb\fto\tya}\ \text{is shapely}
\using\rulefont{ShapeFun}
\end{prooftree}
\\[2ex]
\begin{prooftree}
(x\in\hdenot{}{\mathbb B})
\justifies
x\in\hdenot{}{\mathbb B}\ \text{is shapely}
\using\rulefont{Shape\mathbb B}
\end{prooftree}
\qquad
\begin{prooftree}
(x\in\hdenot{}{\nat})
\justifies
x\in\hdenot{}{\nat}\ \text{is shapely}
\using\rulefont{Shape\nat}
\end{prooftree}
\end{gathered}
$$
\caption{Shapeliness}
\label{fig.shape}
\end{figure}

\begin{lemm}
\label{lemm.shape.back}
\begin{enumerate*}
\item
If $x\in\hdenot{}{\tyb\fto\tya}$ is shapely and $y\in\hdenot{}{\tyb}$ is shapely, then so is $xy\in\hdenot{}{\tya}$.
\item
If $x\in\hdenot{}{[\tya_i]\tya}$ is shapely then $x=\hdenot{\varnothing}{\hd(x)}$.
\item
Every $f\in\nat^\nat$ is shapely.
\end{enumerate*}
\end{lemm}
\begin{proof}
The first two parts follow from the form of the inductive definition in Figure~\ref{fig.shape}.
The third part is a simple application of \rulefont{ShapeFun}, noting that by \rulefont{Shape\mathbb N} every $n\in\nat$ is shapely. 
\end{proof}

We can combine Lemmas~\ref{lemm.shape.back} and~\ref{lemm.unpack.routine} to get a nice corollary of shapeliness ($\f{unpack}$ is from Subsection~\ref{subsect.code.to.values}):
\begin{corr}
If $x\in\hdenot{}{[\tya_i]A}$ is shapely then $\tl(x)=\hdenot{\varnothing}{\f{unpack}\,\hd(x)}$.
\end{corr}
\begin{proof}
Suppose $x\in\hdenot{}{[\tya_i]A}$ is shapely, so that by part~2 of Lemma~\ref{lemm.shape.back} $x=\hdenot{\varnothing}{\hd(x)}$.
We apply $\tl$ to both sides and use Lemma~\ref{lemm.unpack.routine}.
\end{proof}

\begin{lemm}
\label{lemm.sub.X.c}
Suppose $\Gamma,X:[\tyb_i]\tyb\cent r:\tya$, $\Gamma\cent [a_i{:}\tyb_i]s:[\tyb_i]\tyb$, and $\Gamma\cent\varsigma$.
Then $\hdenot{\varsigma}{r[X{\ssm}[a_i{:}\tyb_i]s]}=\hdenot{\varsigma[X{\ssm}\hdenot{\varsigma}{[\tyb_i]s}]}{r}$.
\end{lemm}
\begin{proof}
By a routine induction on the derivation of $\Gamma\cent r:\tya$, similar to the proof of Lemma~\ref{lemm.Xsub.denot}.
\end{proof}

\begin{corr}
\label{corr.shape.varsigmaun}
Suppose $\Gamma\cent r:\tya$,\ \ $\Gamma\cent\varsigma$,\ \ and $\varsigma$ is shapely.
Then $\hdenot{\varsigma}{r}=\hdenot{\varsigma|_\atoms}{r\,\varsigma_\unknowns}$.
\end{corr}
\begin{proof}
First, we note that the effect of $\varsigma_\unknowns$ can be obtained by concatenating $[X{\ssm}\hd(\varsigma(X))]$ for every $X\in\fv(r)$.
The order does not matter because by construction $\hd(\varsigma(X))$ is closed syntax (no free variables).
Furthermore since $\varsigma$ is shapely, $\varsigma(X)=\hdenot{\varnothing}{\hd(\varsigma(X))}$ so we can write $\varsigma$ as 
$$
\varsigma|_\atoms\cup [X{\ssm}\hdenot{\varsigma}{\hd(\varsigma(X))} \mid X\in\dom(\varsigma)],
$$
where here ${[X{\ssm} x_X \mid X\in \mathcal X]}$ is the map taking $X$ to $x_X$ for every ${X\in\mathcal X}$.\footnote{Strictly speaking we also need a version of Proposition~\ref{prop.ws.denote} for the contextual system; this is not hard.}
We now use Lemma~\ref{lemm.sub.X.c} for $[X{\ssm}\varsigma(X)]$ for each $X\in\fv(r)$, and Proposition~\ref{prop.ws.c}.
\end{proof}

\begin{prop}
\label{prop.shape.r}
Suppose $\Gamma\cent r:\tya$ and suppose $\Gamma\cent\varsigma$.
Then if $\varsigma$ is shapely then so is $\hdenot{\varsigma}{r}$.
\end{prop}
\begin{proof}
By induction on the typing $\Gamma\cent r:\tya$ (Figure~\ref{fig.cmtt.types}).
\begin{itemize*}
\item
\emph{The case of \rulefont{Hyp}} \quad is immediate because by assumption $\varsigma(a)$ is shapely.
\item
\emph{The case of \rulefont{Const}} \quad is also immediate (provided that all semantics for constants are shapely).
\item
\emph{The case of \rulefont{{\fto}I}.}\quad
Suppose $\Gamma,a{:}\tya\cent r:\tyb$ so that by \rulefont{{\fto}I} $\Gamma\cent \lam{a{:}\tya}r:\tya\fto\tyb$.
Suppose $x\in\hdenot{}{\tya}$ is shapely.
Then so is $\varsigma[a{\ssm}x]$ and by inductive hypothesis so is $\hdenot{\varsigma[a{\ssm}x]}{r}$.
It follows by \rulefont{ShapeFun} that 
$$
\hdenot{\varsigma}{\lam{a{:}x}r}=\bigl(x\in\hdenot{}{\tya}\mapsto \hdenot{\varsigma[a{\ssm}x]}{r}\bigr)
$$
is shapely.
\item
\emph{The case of \rulefont{{\fto}E}.}\quad
Suppose $\Gamma\cent r':\tya\fto\tyb$ and $\Gamma\cent r:\tya$ so that by \rulefont{{\fto}E} $\Gamma\cent r'r:\tyb$.
By inductive hypothesis $\hdenot{\varsigma}{r'}$ and $\hdenot{\varsigma}{r}$ are both shapely.
By part~1 of Lemma~\ref{lemm.shape.back} so is $\hdenot{\varsigma}{r'r}=\hdenot{\varsigma}{r'}\hdenot{\varsigma}{r}$.
\item
\emph{The case of \rulefont{[\,] I}.}\quad
Suppose $\Gamma,(a_i{:}\tya_i)\cent r:\tya$ and $\fa(r)\subseteq\{a_i\}$ so that by \rulefont{[\,]I} $\Gamma\cent [a_i{:}\tya_i]r:[\tya_i]\tya$.

By inductive hypothesis $\hdenot{\varsigma'}{r}$ is shapely for every shapely $\varsigma'$ such that $\Gamma,(a_i{:}\tya_i)\cent\varsigma'$ and it follows that $\tl\hdenot{\varsigma}{[a_i{:}\tya_i]r}=\hdenot{\varsigma}{\lam{(a_i{:}\tya_i)}r}$ is shapely. 

Also unpacking definitions
$$
\hd\hdenot{\varsigma}{[a_i{:}\tya_i]r}=[a_i{:}\tya_i](r\varsigma_\unknowns).
$$
So it suffices to verify that $\hdenot{\varsigma}{[a_i{:}\tya_i]r}=\hdenot{\varnothing}{[a_i{:}\tya_i](r\varsigma_\unknowns)}$.
This follows from Corollary~\ref{corr.shape.varsigmaun}.
\end{itemize*}
\end{proof}

Corollary~\ref{corr.later.c} proved that denotations cannot be reified to syntax in general, by general arguments on cardinality.
But our denotational semantics is inflated; $\hdenot{}{[\,](\tya\fto\tyb)}$ and $\hdenot{}{\tya\fto\tyb}$ have the same cardinality even if $\hd(\hdenot{}{[\,](\tya\fto\tyb)})$ and $\hdenot{}{\tya\fto\tyb}$ do not.
Corollary~\ref{corr.more.c} tells us that we cannot in general even reify denotation to the `inflated' denotations, even if they are large enough.
In this sense, inflation is `not internally detectable':
\begin{corr}
\label{corr.more.c}
\begin{enumerate*}
\item
There is no term $s$ such that $\varnothing\cent s:(\nat\fto\nat)\fto[\,](\nat\fto\nat)$ is typable and such that 
$\hdenot{\varnothing}{s}\in\hdenot{}{[\,](\nat\fto\nat)}^{\hdenot{}{\nat\fto\nat}}$ is injective.
\item
There is no term $s$ such that $\varnothing\cent s:(\nat\fto\nat)\fto[\nat]\nat$ is typable and such that
$\hdenot{\varnothing}{s}\in\hdenot{}{[\nat]\nat}^{\hdenot{}{\nat\fto\nat}}$ is injective.
\end{enumerate*}
\end{corr}
\begin{proof}
By Proposition~\ref{prop.shape.r} $s$ is shapely, so by part~1 of Lemma~\ref{lemm.shape.back} it maps shapely elements of $\nat^\nat=\hdenot{}{\nat\fto\nat}$ to shapely elements of $\hdenot{}{[\,](\nat\fto\nat)}$/$\hdenot{}{[\nat]\nat}$.
By part~3 of Lemma~\ref{lemm.shape.back} and the fact that $\nat^\nat$ is uncountable, the number of shapely elements of $\nat^\nat$ is uncountable. 
By part~2 of Lemma~\ref{lemm.shape.back} and the fact that syntax is countable, the number of shapely elements of $\hdenot{}{[\,](\nat\fto\nat)}$ and $\hdenot{}{[\nat]\nat}$ is countable.
The result follows.
\end{proof}

It is clear that part~1 of Corollary~\ref{corr.more.c} can be directly adapted to the modal system from Section~\ref{sect.box.syntax}.

\section{$\Box$ as a (relative) comonad}
\label{sect.box.comonad}

We noted as early as Remark~\ref{rmrk.modal} that $\Box$ looks like a comonad.
In this section, we show that this is indeed the case.

Before doing this, we would like to convince the reader that this is obviously impossible.

True, we have natural maps $\Box\tya\to\tya$ (evaluation) and $\Box\tya\to\Box\Box\tya$ (quotation).
However, if $\Box$ is a comonad then it has to be a functor on some suitable category, so we would expect some natural map in $(\tya\fto\tyb)\to(\Box\tya\fto\Box\tyb)$.
This seems unlikely because if we had this, then we could take $\tya$ to be a unit type (populated by one element) and $\tyb=(\nat\fto\nat)$ and thus generate a natural map from $\nat\fto\nat$ to $\Box(\nat\fto\nat)$.
But how would we do this in the light of Corollaries~\ref{corr.later.c} and Corollary~\ref{corr.more.c}? 
Even where closed syntax exists for a denotation, there may be many different choices of closed syntax to represent the same denotation, further undermining our chances of finding \emph{natural} assignments.
`$\Box$ as a comonad' seems doomed. 

This problem is circumvented by the `trick' of considering a category in which each denotation must be associated with syntax; we do not insist that the syntax and denotation match.
This is essentially the same idea as \emph{inflation} in Remark~\ref{rmrk.how.to.interpret.types} (but applied in the other direction; in Remark~\ref{rmrk.how.to.interpret.types} we inflated by adding a purported denotation to every syntax; here we are inflating by adding a purported syntax to every denotation).
In the terminology of Definition~\ref{defn.shape} we can say that we do not insist on \emph{shapeliness}.
We simply insist that some syntax be provided.

Modulo this `trick', $\Box$ becomes a well-behaved comonad after all.

\subsection{$\Box$ as a comonad}

\begin{nttn}
Write $\pi_1$ for \emph{first projection} and $\pi_2$ for \emph{second projection}.

That is, $\pi_1(x,y)=x$ and $\pi_2(x,y)=y$.
\end{nttn}

\begin{defn}
\label{defn.Box.f}
Suppose $f\in\hdenot{}{\Box\tyb}^{\hdenot{}{\Box\tya}}$.
Define a function $\BBox f\in\hdenot{}{\Box\Box\tyb}^{\hdenot{}{\Box\Box\tya}}$ by sending 
$$
\Box\Box s::x
\quad\text{to}\quad
\Box\pi_1(f(\Box s::\hdenot{\varnothing}{s}))::f(x) 
$$
where $x\in\hdenot{}{\Box\tya}$ and $s:\tya$.
\end{defn}

\begin{rmrk}
It may be useful to unpack what $\BBox f$ does. 
Suppose 
$$
f(\Box r::x)=\Box r'::x'
\quad\text{and}\quad
f(\Box s::y)=\Box s'::y'
$$ 
where $x\in\hdenot{}{\tya}$ and $y=\hdenot{\varnothing}{s}$.
Then $\Box f$ sends $\Box\Box s::\Box r::x$ to $\Box\Box s'::\Box r'::x'$.
\end{rmrk}

\begin{defn}
\label{defn.J}
Define a category $\mathcal J$ by:
\begin{itemize*}
\item
Objects are types $\tya$.\footnote{The reader might prefer to take objects to be $\hdenot{}{\tya}$.  This is fine; the assignment $\tya\longmapsto\hdenot{}{\tya}$ is injective, so it makes no difference whether we take objects to be $\tya$ or $\hdenot{}{\tya}$.}
\item
Arrows from $\tya$ to $\tyb$ are functions from $\hdenot{}{\Box\tya}$ to $\hdenot{}{\Box\tyb}$ (not from $\hdenot{}{\tya}$ to $\hdenot{}{\tyb}$; as promised above, some syntax must be provided).
\end{itemize*}
Composition of arrows is given by composition of functions.
\end{defn}

\begin{defn}
\label{defn.Box.functor}
Define an endofunctor $\BBox$ on $\mathcal J$ mapping 
\begin{itemize*}
\item
an object $\tya$ to $\BBox\tya=\Box\tya$ and 
\item
an arrow $f:\tya\to\tyb$ to $\BBox f:\BBox\tya\to\BBox\tyb$ from Definition~\ref{defn.Box.f}.
\end{itemize*}
\end{defn}

\noindent So $\Box$ is a type-former acting on types and $\BBox$ is a functor acting on objects and arrows.
Objects happen to be types, and $\BBox$ acts on objects just by prepending a $\Box$.
Arrows are functions on sets, and the action on $\BBox$ on these functions is more complex as defined above.

\begin{defn}
\label{defn.delta.tya}
\begin{itemize*}
\item
Write $\id_\tya$ for the identity on $\hdenot{}{\Box\tya}$ for each $\tya$.
\item
Write $\delta_\tya$ for the arrow from $\BBox\tya$ to $\tya$ given by the function mapping $\hdenot{}{\Box\Box\tya}$ to $\hdenot{}{\Box\tya}$ taking $\Box\Box r::x$ to $x$ (where $x\in\hdenot{}{\Box\tya}$).
This will be the \deffont{counit} of our comonad.
\item
Write $\epsilon_\tya$ for the arrow from $\BBox\tya$ to $\BBox\BBox\tya$ given by the function mapping $\hdenot{}{\Box\Box\tya}$ to $\hdenot{}{\Box\Box\Box\tya}$ taking $\Box\Box r::x$ to $\Box\Box\Box r::\Box\Box r::x$ (where $x\in\hdenot{}{\Box\tya}$).
This will be the \deffont{comultiplication} of our comonad.
\end{itemize*}
\end{defn}

\begin{lemm}
$\BBox$ from Definition~\ref{defn.Box.functor} is a functor.
\end{lemm}
\begin{proof}
It is routine to verify that $\BBox\id_\tya=\id_{\BBox\tya}$ and if $f:\tya\to\tyb$ and $g:\tyb\to\tyc$ then $\BBox g\circ\BBox f=\BBox(g\circ f)$.
\end{proof}

\begin{lemm}
\begin{itemize*}
\item
$\delta_\tya$ is a natural transformation from $\BBox$ to $\id_{\mathcal J}$ (the identity functor on $\mathcal J$). 
\item
$\epsilon_\tya$ is a natural transformation from $\BBox$ to $\BBox\BBox$.
\end{itemize*}
\end{lemm}
\begin{proof}
Suppose $f:\tya\to\tyb$.
For the first part, we need to check that $f\circ\delta_\tya = \delta_\tyb\circ \BBox f$.
This is routine: 
$$
\begin{array}{l}
(f\circ\delta_\tya)(\Box\Box r::x)=f(x)
\quad\text{and}
\\
\delta_\tyb\circ\BBox f=\pi_2\bigl(\Box\pi_1(f(\Box r::\hdenot{\varnothing}{r}))::f(x)\bigr)=f(x)
\end{array}
$$
The second part is similar and no harder.
\end{proof}

Note that $\BBox\delta_\tya$ is an arrow from $\BBox\BBox\tya$ to $\BBox\tya$.
\begin{lemm}
\label{lemm.Box.delta.a}
$\BBox\delta_\tya$ maps $\Box\Box\Box s::\Box\Box r::x\in\hdenot{}{\Box\Box\Box\tya}$ to $\Box\Box s::x\in\hdenot{}{\Box\Box\tya}$.
\end{lemm}
\begin{proof}
By a routine calculation on the definitions:
$$
\begin{array}{r@{\ }l@{\quad}l}
\Box\delta_\tya(\Box\Box\Box s::\Box\Box r::x)
=&\Box\pi_1(\delta_\tya(\Box\Box\Box s::\hdenot{\varnothing}{\Box\Box s}))::\delta_\tya(\Box\Box r::x)
&\text{Definition~\ref{defn.Box.f}}
\\
=&\Box\pi_1(\delta_\tya(\Box\Box\Box s::\hdenot{\varnothing}{\Box\Box s}))::x
&\text{Definition~\ref{defn.delta.tya}}
\\
=&\Box\pi_1(\hdenot{\varnothing}{\Box\Box s})::x
&\text{Definition~\ref{defn.delta.tya}}
\\
=&\Box\Box s::x
&\text{Figure~\ref{fig.denot.terms}}
\end{array}
$$
\end{proof}
 
\begin{prop}
$\BBox$ is a comonad.
\end{prop}
\begin{proof}
We need to check that 
\begin{itemize*}
\item
$\BBox\epsilon_\tya\circ\epsilon_\tya=\epsilon_{\BBox\tya}\circ\epsilon_{\tya}$ and 
\item
$\delta_{\BBox\tya}\circ\epsilon_\tya=\id_\tya=\BBox\delta_\tya\circ\epsilon_\tya$.
\end{itemize*}
Both calculations are routine.
We consider just the second one.
Consider $\Box\Box s::\Box r::x\in\hdenot{}{\Box\Box\tya}$.
Then
$$
\begin{array}{r@{\ }l}
(\delta_{\BBox\tya}\circ\epsilon_\tya)(\Box\Box s::\Box r::x)=&\delta_{\BBox\tya}(\shademath{\Box\Box\Box s::}\Box\Box s::\Box r::x)
\\
=&\Box\Box s::\Box r::x
\\[1.5ex]
(\BBox\delta_\tya\circ\epsilon_\tya)(\Box\Box s::\Box r::x)=&\BBox\delta_\tya(\Box\Box\Box s::\shademath{\Box\Box s::}\Box r::x)
\\
=&\Box\Box s::\Box r::x
\end{array}
$$
The shaded part is the part that gets `deleted'. 
In the second case we use Lemma~\ref{lemm.Box.delta.a}. 
\end{proof}

\subsection{$\Box$ as a relative comonad}

Recall that in the previous subsection we represented $\Box$ as a comonad on a category with the `trick' of associating syntax to every denotation.

It is possible to put this in a broader context using the notion of \emph{relative comonad}.

\begin{defn}
Following \cite{chapman:monnne}, a \deffont{relative comonad} consists of the following information:
\begin{itemize*}
\item
Two categories $\mathcal J$ and $\mathcal C$ and a functor $J:\mathcal J\to\mathcal C$.\footnote{The clash with the $\mathcal J$ from Definition~\ref{defn.J} is deliberate: this is the only $\mathcal J$ we will care about in this paper.
The definition of relative comonad from \cite{chapman:monnne} is general in the source category.}
\item
A functor $T:\mathcal J\to\mathcal C$.
\item
For every $X\in\mathcal J$ an arrow $\delta_X:TX\to JX\in\mathcal C$ (the \deffont{unit}).
\item
For every $X,Y\in\mathcal J$ and arrow $k:TX\to JY\in\mathcal C$, an arrow $k^*:TX\to TY$ (the \deffont{Kleisli extension}). 
\end{itemize*}
Furthermore, we insist on the following equalities:
\begin{itemize*}
\item
If $X,Y\in\mathcal J$ and $k:X\to Y\in\mathcal J$ then $k=k^*\circ\delta$. 
\item
If $X\in\mathcal J$ then $\delta_X^*=\id_{TX}$.
\item
If $X,Y,Z\in\mathcal J$ and $k:TX\to JY$ and $l:TY\to JZ$ then $l^*\circ k^*=(l\circ k)^*$. 
\end{itemize*}
\end{defn}

\begin{defn}
\label{defn.relative.comonad}
Take $\mathcal C$ to have objects types $\tya$ and arrows elements of $\hdenot{}{\tyb}^\hdenot{}{\tya}=\hdenot{}{\tya\fto\tyb}$---this is simply the natural category arising from the denotational semantics of Figure~\ref{fig.denot.types}.

Take $\mathcal J$ to be the category of Definition~\ref{defn.J}.

Take $J$ to map $\tya\in\mathcal J$ to $\Box\tya\in\mathcal C$ and to map $f\in\hdenot{}{\Box\tyb}^{\hdenot{}{\Box\tya}}$ to itself.

Take $T$ to map $\tya\in\mathcal J$ to $\Box\Box\tya\in\mathcal C$ and to map $f\in\hdenot{}{\Box\tyb}^{\hdenot{}{\Box\tya}}$ to $\BBox f$ from Definition~\ref{defn.Box.f}.
\end{defn}

\begin{prop}
Definition~\ref{defn.relative.comonad} determines a relative comonad on $\mathcal C$.
\end{prop}

It is slightly simplified, but accurate, to describe relative (co)monads as being for the case where we have an operator that is nearly (co)monadic but the category in question has `too many objects'.
By that view, $\Box$ is a comonad on the full subcategory of $\mathcal C$ over modal types.

Now the intuition of modal types $\Box\tya$ is `closed syntax', so it may be worth explicitly noting here that this full subcategory is \emph{not} just a category of syntax.
Each $\hdenot{}{\Box\tya}$ contains for each term $\varnothing\cent r:\tya$ also a copy of $\hdenot{}{\tya}$, because we inflate.

\section{Conclusions}

The intuition realised by the denotation of $\Box\tya$ in this paper means `typable closed syntax of the same language, of type $\tya$'.
This is difficult to get right because it is self-referential;
if we are careless then the undecidable runtime impinges on the inductively defined denotation.
We noted this in Subsection~\ref{subsubsect.explain}.

For that reason we realised this intuition by an `inflated' reading of $\Box\tya$ as `closed syntax, and purported denotation of that syntax'.
As noted in Remark~\ref{rmrk.how.to.interpret.types}, there is no actual restriction that $\Box r::x\in\hdenot{}{\Box\tya}$ needs to match up, in that $r$ must have denotation $x$.

When $r$ and $x$ \emph{do} match up we say that $\Box r::x$ is \emph{shapely}.
This is Definition~\ref{defn.shape}, and we use this notion for our culminating result in Corollary~\ref{corr.more.c}, which entails that there is no unform family of terms of type $\tya\fto\Box\tya$.

The proof of this involves a beautiful interplay between syntax and denotation, 
which also illustrates the usefulness of denotational techniques; we can use a sound model to show that certain things \emph{cannot} happen in the syntax, because if they did, they would have to happen in the model.

\paragraph{Future work}

One avenue for future work is to note that our denotation is \emph{sets based}, and so this invites generalisation to \emph{nominal} sets semantics \cite{gabbay:newaas-jv}.

Perhaps we could leverage this to design a language which combines the simplicity of the purely modal system with the expressivity of contextual terms.
Specifically, nominal sets are useful for giving semantics to open terms \cite{gabbay:nomhss,gabbay:stodfo} and we hope to develop a language in which we can retain the modal type system but relax the condition that $\fa(r)=\varnothing$ in \rulefont{\Box I} in Figure~\ref{fig.modal.types} (much as the contextual system does, but in the `nominal' approach we would not add types to the modality).

The underlying motivation here is that the contextual system is `eager' in accounting for free variables---we need to express all the variables we intend to use in the contextual modal type, by putting their types in the modality. 
We might prefer to program on open syntax in a `lazy' fashion, by stating that the syntax may be open, but not specifying its free variables explicitly in the type.

Note that this is not the same thing as programming freely on open syntax.
Free variables would still be accounted for in the typing context (leading to some form of \emph{dynamic linking} as and when open syntax is unboxed and evaluated; for an example of a $\lambda$-calculus view of dynamic linking, though not meta-programming, see \cite{ancona:caldl}).
So all variables would be eventually accounted for in the typing context, but they would not need to be listed in the type. 

This is another reason for the specific design of our denotional semantics and taking the denotation of $\Box\tya$ to be specifically closed syntax; we hope to directly generalise this using nominal techniques so that $\Box\tya$ can also denote (atoms-)open syntax.
This is future work.

\paragraph{On the precise meaning of Corollary~\ref{corr.more.c}}

Corollary~\ref{corr.more.c} depends on the fact that we admitted no constants of type $(\nat\fto\nat)\fto\Box(\nat\fto\nat)$.
We may be able to admit such a constant, representing a function that takes denotation and associates to it some `dummy syntax' chosen in some fixed but arbitrarly manner.

So Corollary~\ref{corr.more.c} does not (and should not) prove that terms of type $(\nat\fto\nat)\fto\Box(\nat\fto\nat)$ are completely impossible---only that they do not arise from the base system and cannot exist unless we explicitly choose to put them in there.

\paragraph{Technical notes on the jump in complexity from modal to contextual system}

We noted in the introduction that Sections~\ref{sect.box.syntax} and~\ref{sect.contextual.system}, and Sections~\ref{sect.denotations} and~\ref{sect.contextual.models} are parallel developments of the syntax and examples of the modal and contextual systems. 

We briefly survey technical details of how these differences manifest themselves.
\begin{itemize*}
\item
The contextual system enriches the modal system with types in the modality.
The increase in expressivity is exemplified in Subsection~\ref{subsubsect.exp.c}.
\item
In the contextual system and not in the modal system, instantiation of unknowns can trigger an atoms-substitution (see Definition~\ref{defn.sub.c}) leading to a kind of `cascade effect'. 
This turns out to be terminating, well-behaved, and basically harmless---but this has to be verified, and that brings some specific technical material forward in the proofs for the contextual case that is not so prominent in the purely modal case (notably, Lemma~\ref{lemm.sigma.c}).
\item
A clear view of exactly where the extra complexity of the contextual system `lives' in the denotation can be obtained by comparing the denotational semantics of $\Box\tya$ and $[\tya_i]\tya$ in Figures~\ref{fig.denot.types} and~\ref{fig.denot.types.c}.
\end{itemize*}

\subsection*{Related work}

\paragraph{$\Box$ and monads}

Famously, Moggi proposed to model computation using a monad \cite{moggi:notcm}.
Let us write it as $\Diamond\tya$.\footnote{Pfenning and Davies discuss this in \cite[Section~7, page~21]{pfe+dav:mscs01}.}
This type is intuitively populated by `computations of type $\tya$'.
The unit arrow $\tya\fto\Diamond\tya$ takes a value of type $\tya$ and returns the trivial computation that just returns $\tya$.

The difference from the \emph{co}monad of this paper in that our $\Box\tya$ is populated by \emph{closed syntax}, and not by \emph{computation}.

If we have an element of $\nat^\nat$ then it is easy to build a computation that just returns that value; it is however not easy---and may be impossible---to exhibit closed syntax to represent this computation.

We could add a constant to our syntax for each of the uncountably many functions from natural numbers to natural numbers.
This would be mathematically fine---but not particularly implementable.
We do not assume this.

Closed syntax is of course related to computation, and we can make this formal:
Given an element in $\Box(\nat\fto\nat)$ we can map it to a computation, just by executing it.
So intuitively there is an arrow $\Box\tya\to\Diamond\tya$.
In the modal logic tradition this is called axiom \rulefont{D}.

In summary: we propose that the Moggi-style monads corresponds to a modal $\Diamond$, whereas CMTT-style $\Box$ is a modal $\Box$ and corresponds to a comonadic structure. 

See also~\cite{kob:tcs97,bie+pai:sl00,ale+men+pai+rit:csl01}, where the
$\Box$ operator of several constructive variants of S4 (not equivalent to the version we presented here) is modeled as comonads.

\paragraph{Brief survey of applications of $\Box$ calculi}

Logic and denotation, not implementation, are the focus of this paper, but the `$\Box$-calculi' considered in this paper have their motivation in implementation and indeed they were specifically designed to address implementational concerns.
We therefore give a brief survey of how (contextual) modal types have been useful in the more applied end of computer science.

The connection of the modal $\Box$ calculus with partial evaluation and staged computation was noticed by Davies and Pfenning~\cite{dav+pfe:acm01,pfe+dav:mscs01}, and subsequently used as a language for run-time code generation by Wickline et
al.~\cite{wic+lee+pfe:pldi98}. 
The contextual variant of $\Box$ as a
basis for meta-programming and modeling of higher-order abstract
syntax was proposed by Nanevski and
Pfenning~\cite{nan+pfe:jfp05}, and subsequently used to reason
about optimised implementation of higher-order unification in
Twelf~\cite{pie+pfe:cade03}, which could even be scaled to dependent
types~\cite{nan+pfe+pie:tcl08}.  

Recently, the contextual flavor of
the system has been used in meta-programming applications for
reasoning and programming with higher-order abstract syntax by Pientka and collaborators~\cite{pie:popl08,pie+dun:ppdp08,fel+pie:itp10,cav+pie:popl12}.

\paragraph{Relationship between the formulation with meta-variables and labeled natural deductions}

The syntax of terms from Definition~\ref{defn.terms} does not follow instantly from the syntax of types from Definition~\ref{defn.types}; in particular, the use of a two-level syntax (also reminiscent of the two levels of nominal terms \cite{gabbay:nomu-jv}) is a design choice, not an inevitability. 

The usual way to present inference systems based on modal logic is
to have a propositional (or variable) context where each proposition
is labeled by the `world' at which it is
true~\cite{simpson:phd94}. 

When S4 is considered, we take advantage of reflexivity and
transitivity of the Kripke frame to simplify the required information to two kinds of facts:
\begin{enumerate*}
\item
What holds at the current world, but not necessarily in all future worlds.
\item
What holds in the current world and also in all future worlds.
\end{enumerate*} 

By this view, the first kind of fact corresponds to atoms $a$, and the second kind of fact corresponds to unknowns $X$.
So this can be seen as the origin of the two-level structure of our syntax in this paper. 

The interested reader can find the modal (non-contextual) version of our type-system presented using the labeled approach in a paper by Davies and Pfenning~\cite{dav+pfe:acm01}, and each stage of computation is indeed viewed as world in a Kripke frame.

\paragraph{CMTT and nominal terms}

Nominal terms were developed in \cite{gabbay:nomu,gabbay:nomu-jv} and feature a two-level syntax, just like CMTT.
That is made very clear in this paper, where the first author imported the nominal terms terminology of \emph{atoms} and \emph{unknowns}.

The syntax of this paper is not fully nominal---the $[a_i]r$ of the contextual system may look like a nominal abstraction, but there are no suspended permutation $\pi\act X$ (instead, we have types in the modality).
One contribution of this paper is to make formal, by a denotation, the precise status of the two levels of variable in CMTT.

So we can note that the abstraction for atoms is functional abstraction in CMTT whereas the abstraction for atoms in nominal terms is nominal atoms-abstraction;\footnote{In \cite{gabbay:unialt} we translate nominal terms to higher-order terms, and atoms-abstraction gets translated to functional abstraction.  However, this does \emph{not} mean that atoms-abstraction is a `special case' of functional abstraction, any more than translating e.g. Java to machine binary means that method invocation is a special case of logic gates.} unknowns of nominal terms range over elements of nominal sets, whereas unknowns of CMTT range over ordinary sets functionally abstracted over finitely many arguments; the notion of \emph{equivariance} (symmetry up to permuting atoms) characteristic of all nominal techniques is absent in CMTT (the closest we get is a term like $\f{exchange}_{B,C}$ in Subsection~\ref{subsect.modal-style.axioms}); and in contrast the self-reflective character of CMTT is absent from nominal terms and the logics built out of it \cite{gabbay:nomtnl}. 
So in spite of some structural parallels between CMTT and nominal terms in that both are two-level, there are also significant differences.

As noted above, there is a parallel between CMTT and Kripke structures, that is made more explicit in \cite{dav+pfe:acm01}.
A direct connection between nominal terms and Kripke semantics has never been made, but the first author at least has been aware of it as a possibility, where `future worlds' corresponds to `more substitutions arriving'.
Also as discussed above, an obvious next step is to develop a modified modal syntax which takes on board more `nominal' ideas, applied to the modal intuitions which motivate the $\lambda$-calculus of this paper.
This is future work.

\paragraph{The syntax of this paper, and previous work}

The modal and contextual systems which we give semantics to in this
paper, are taken from previous work.  Specifically,
Definition~\ref{defn.terms} corresponds to \cite{pfe+dav:mscs01}, 
Definition~\ref{defn.terms.c} corresponds to \cite{nan+pfe+pie:tcl08},
Figure~\ref{fig.modal.types} corresponds to \cite{pfe+dav:mscs01} and
Figure~\ref{fig.cmtt.types} to \cite{nan+pfe+pie:tcl08}.

We cannot give specific definition references in the citations to \cite{nan+pfe+pie:tcl08} and \cite{pfe+dav:mscs01}, because those papers never give a specific definition of their syntax.  
If they did, then they would correspond as described. 
We do feel that this paper does make some contribution in terms
of presentation, and the exposition and definitions here may be tailored to a slightly different community.

\section*{Acknowledgements}
\noindent This paper was supported by Spanish MICINN Project TIN2010-20639 Paran10; AMAROUT grant PCOFUND-GA-2008-229599; Ramon y Cajal grants RYC-2010-0743 and RYC-2006-002131; and the Leverhulme Trust.


\begin{thebibliography}{AMdPR01}

\bibitem[ACU10]{chapman:monnne}
Thorsten Altenkirch, James Chapman, and Tarmo Uustalu.
\newblock Monads need not be endofunctors.
\newblock In {\em Foundations of software science and computation structures,
  13th International Conference ({FOSSACS} 2010)}, volume 6014 of {\em Lecture
  Notes in Computer Science}, pages 297--311. Springer, 2010.

\bibitem[AFZ03]{ancona:caldl}
Davide Ancona, Sonia Fagorzi, and Elena Zucca.
\newblock A calculus for dynamic linking.
\newblock In {\em ICTCS}, pages 284--301, 2003.

\bibitem[AL91]{asperti:cattsi}
Andr\'ea Asperti and Giuseppe Longo.
\newblock {\em Categories, types, and structures: an introduction to category
  theory for the working computer scientist}.
\newblock Foundations of computing. {MIT} Press, 1991.
\newblock Available online from the University of Michigan, digitised November
  2007.

\bibitem[AMdPR01]{ale+men+pai+rit:csl01}
Natasha Alechina, Michael Mendler, Valeria de~Paiva, and Eike Ritter.
\newblock Categorical and {K}ripke semantics for {C}onstructive {S4} modal
  logic.
\newblock In {\em Computer Science Logic, {CSL}'01}, volume 2142 of {\em
  Lecture Notes in Computer Science}, pages 292--307, 2001.

\bibitem[BdP00]{bie+pai:sl00}
Gavin~M. Bierman and Valeria C.~V. de~Paiva.
\newblock On an intuitionistic modal logic.
\newblock {\em Studia Logica}, 65(3):383--416, 2000.

\bibitem[BdRV01]{blackburn:modl}
Patrick Blackburn, Maarten de~Rijke, and Yde Venema.
\newblock {\em Modal Logic}.
\newblock Cambridge University Press, 2001.

\bibitem[CP12]{cav+pie:popl12}
Andrew Cave and Brigitte Pientka.
\newblock Programming with binders and indexed data-types.
\newblock In {\em Proceedings of the 39th {ACM} {SIGPLAN-SIGACT} Symposium on
  {P}rinciples of {P}rogramming {L}anguages ({POPL}'12)}. ACM, 2012.
\newblock accepted.

\bibitem[DP01]{dav+pfe:acm01}
Rowan Davies and Frank Pfenning.
\newblock A modal analysis of staged computation.
\newblock {\em Journal of the {ACM}}, 48(3):555--604, 2001.

\bibitem[FP10]{fel+pie:itp10}
Amy Felty and Brigitte Pientka.
\newblock Reasoning with higher-order abstract syntax and contexts: A
  comparison.
\newblock In {\em Interactive Theorem Proving}, volume 6172 of {\em Lecture
  Notes in Computer Science}, pages 227--242, 2010.

\bibitem[Gab11]{gabbay:stodfo}
Murdoch~J. Gabbay.
\newblock \href{http://www.gabbay.org.uk/papers.html\#stodfo}{Stone duality for
  First-Order Logic: a nominal approach}.
\newblock In {\em Howard Barringer Festschrift}. December 2011.

\bibitem[Gab12]{gabbay:nomtnl}
Murdoch~J. Gabbay.
\newblock \href{http://www.gabbay.org.uk/papers.html\#nomtnl}{Nominal terms and
  nominal logics: from foundations to meta-mathematics}.
\newblock In {\em Handbook of Philosophical Logic}, volume~17. Kluwer, 2012.

\bibitem[GKWZ03]{gabbay:mandml}
Dov~M. Gabbay, Agnes Kurucz, Frank Wolter, and Michael Zakharyaschev.
\newblock {\em Many-dimensional modal logics: theory and applications}, volume
  148 of {\em Studies in Logic and the Foundations of Mathematics}.
\newblock Elsevier, 2003.

\bibitem[GM09]{gabbay:unialt}
Murdoch~J. Gabbay and Dominic~P. Mulligan.
\newblock \href{http://www.gabbay.org.uk/papers/unialt.pdf}{Universal algebra
  over lambda-terms and nominal terms: the connection in logic between nominal
  techniques and higher-order variables}.
\newblock In {\em Proceedings of the 4th International Workshop on Logical
  Frameworks and Meta-Languages ({LFMTP} 2009)}, pages 64--73. ACM, August
  2009.

\bibitem[GM11]{gabbay:nomhss}
Murdoch~J. Gabbay and Dominic Mulligan.
\newblock \href{http://www.gabbay.org.uk/papers.html\#nomhss}{Nominal {H}enkin
  {S}emantics: simply-typed lambda-calculus models in nominal sets}.
\newblock In {\em Proceedings of the 6th International Workshop on Logical
  Frameworks and Meta-Languages ({LFMTP} 2011)}, volume~71 of {\em EPTCS},
  pages 58--75, September 2011.

\bibitem[GP01]{gabbay:newaas-jv}
Murdoch~J. Gabbay and Andrew~M. Pitts.
\newblock \href{http://www.gabbay.org.uk/papers.html\#newaas-jv}{A New Approach
  to Abstract Syntax with Variable Binding}.
\newblock {\em Formal Aspects of Computing}, 13(3--5):341--363, July 2001.

\bibitem[Kob97]{kob:tcs97}
Satoshi Kobayashi.
\newblock Monad as modality.
\newblock {\em Theoretical Computer Science}, 175(1):29--74, 1997.

\bibitem[Mit96]{mitchell:foupl}
John~C. Mitchell.
\newblock {\em Foundations for Programming Languages}.
\newblock {MIT} Press, 1996.

\bibitem[Mog91]{moggi:notcm}
Eugenio Moggi.
\newblock Notions of computation and monads.
\newblock {\em Information and Computation}, 93(1):55--92, 1991.

\bibitem[NP05]{nan+pfe:jfp05}
Aleksandar Nanevski and Frank Pfenning.
\newblock Staged computation with names and necessity.
\newblock {\em Journal of Functional Programming}, 15(6):893--939, 2005.

\bibitem[NPP08]{nan+pfe+pie:tcl08}
Aleksandar Nanevski, Frank Pfenning, and Brigitte Pientka.
\newblock Contextual modal type theory.
\newblock {\em ACM Transactions on Computational Logic}, 9(3):1--49, 2008.

\bibitem[PD01]{pfe+dav:mscs01}
Frank Pfenning and Rowan Davies.
\newblock A judgmental reconstruction of modal logic.
\newblock {\em Mathematical Structures in Computer Science}, 11(4), 2001.

\bibitem[PD08]{pie+dun:ppdp08}
Brigitte Pientka and Joshua Dunfield.
\newblock Programming with proofs and explicit contexts.
\newblock In {\em Proceedings of the 10th International {ACM} {SIGPLAN}
  Symposium on Principles and Practice of Declarative Programming ({PPDP}
  2008)}, pages 163--173, 2008.

\bibitem[Pie08]{pie:popl08}
Brigitte Pientka.
\newblock A type-theoretic foundation for programming with higher-order
  abstract syntax and first-class substitutions.
\newblock In {\em Proceedings of the 35th {ACM} {SIGPLAN-SIGACT} Symposium on
  {P}rinciples of {P}rogramming {L}anguages ({POPL}'08)}, pages 371--382. ACM,
  2008.

\bibitem[PP03]{pie+pfe:cade03}
Brigitte Pientka and Frank Pfennning.
\newblock Optimizing higher-order pattern unification.
\newblock In {\em Proceedings of the International Conference on Automated
  Deduction ({CADE}'03)}, volume 2741 of {\em Lecture Notes in Computer
  Science}, pages 473--487, 2003.

\bibitem[Sim94]{simpson:phd94}
Alex~K. Simpson.
\newblock {\em The Proof Theory and Semantics of Intuitionistic Modal Logic}.
\newblock PhD thesis, University of Edinburgh, 1994.

\bibitem[UPG03]{gabbay:nomu}
Christian Urban, Andrew~M. Pitts, and Murdoch~J. Gabbay.
\newblock \href{http://www.gabbay.org.uk/papers.html\#nomu}{Nominal
  Unification}.
\newblock In {\em CSL}, volume 2803 of {\em Lecture Notes in Computer Science},
  pages 513--527. Springer, December 2003.

\bibitem[UPG04]{gabbay:nomu-jv}
Christian Urban, Andrew~M. Pitts, and Murdoch~J. Gabbay.
\newblock \href{http://www.gabbay.org.uk/papers.html\#nomu-jv}{Nominal
  Unification}.
\newblock {\em Theoretical Computer Science}, 323(1--3):473--497, September
  2004.

\bibitem[WLP98]{wic+lee+pfe:pldi98}
Philip Wickline, Peter Lee, and Frank Pfenning.
\newblock Run-time code generation and {Modal-ML}.
\newblock In {\em Programming Language Design and Implementation ({PLDI}'98)},
  pages 224--235. ACM, 1998.

\end{thebibliography}

\end{document}